\documentclass[11pt]{article}
\usepackage{algorithm}
\usepackage{algorithmic}
\usepackage{mathpazo}
\usepackage[margin=1in]{geometry}

\usepackage{enumitem}
\usepackage{xcolor}
\usepackage{soul}

\newcommand{\mathcolorbox}[2]{\colorbox{#1}{$\displaystyle #2$}}

\usepackage{amsmath}
\usepackage{amssymb}
\usepackage{amsthm}
\usepackage[]{algorithm}
\usepackage{setspace}
\usepackage{hyperref}
\hypersetup{colorlinks=true,linkcolor=black,citecolor=[rgb]{0.5,0,0}}
\usepackage{graphicx,nicefrac}
\usepackage{subcaption}
\usepackage[font={small}]{caption}
\usepackage{thm-restate}
\usepackage[normalem]{ulem}
\usepackage{float}
\usepackage{balance}
\usepackage{tikz}
\usetikzlibrary{shapes}
\usetikzlibrary{plotmarks}
\usepackage[
        noabbrev,
        capitalise,
        nameinlink,
    ]
    {cleveref}
\usepackage{enumitem}
\usepackage{xcolor}
\newcommand{\nc}{\newcommand}
\newcommand{\DMO}{\DeclareMathOperator}
\DeclareMathAlphabet\mathbfcal{OMS}{cmsy}{b}{n}

\providecommand{\Comments}{1}
\ifnum\Comments=1
\newcommand{\pasin}[1]{\textcolor{red}{[Pasin: #1]}}
\newcommand{\cs}[1]{\textcolor{red}{[Karthik: #1]}}
\newcommand{\todo}[1]{\textcolor{red}{[TODO: #1]}}

\else

\newcommand{\pasin}[1]{}
\newcommand{\cs}[1]{}
\newcommand{\todo}[1]{}
\fi

\renewcommand{\tilde}{\widetilde}
\renewcommand{\hat}{\widehat}
\allowdisplaybreaks

\nc{\MS}{\mathcal{S}}
\nc{\MP}{\mathcal{P}}
\nc{\MR}{\mathcal{R}}
\nc{\cM}{\mathcal{M}}
\nc{\cS}{\mathcal{S}}
\nc{\cI}{\mathcal{I}}
\nc{\cA}{\mathcal{A}}
\nc{\tcA}{\widetilde{\cA}}

\nc{\MZ}{\mathcal{Z}}

\DMO{\Binom}{Binom}
\newcommand{\E}{\mathbb{E}}
\DMO{\Var}{Var}

\newcommand{\bX}{\mathbf{X}}
\renewcommand{\phi}{\varphi}
\nc{\tbx}{\widetilde{\bx}}
\nc{\tbX}{\widetilde{\bX}}
\nc{\tZ}{\widetilde{Z}}
\nc{\tz}{\widetilde{z}}
\newcommand{\bU}{\mathbf{U}}
\nc{\tbU}{\widetilde{\bU}}
\newcommand{\bT}{\mathbf{T}}
\nc{\tbT}{\widetilde{\bT}}
\newcommand{\bD}{\mathbf{D}}
\nc{\tbD}{\widetilde{\bD}}

\newcommand{\bx}{\mathbf{x}}
\newcommand{\NP}{\mbox{\sf NP}}
\renewcommand{\P}{\mbox{\sf P}}
\newcommand{\PSPACE}{\mbox{\sf PSPACE}}
\newcommand{\RIH}{\mbox{\sf RIH}}
\newcommand{\PCP}{\mbox{\sf PCP}}
\newcommand{\PCPP}{\mbox{\sf PCPP}}
\newcommand{\LTC}{\mbox{\sf LTC}}
\newcommand{\CSP}{\mbox{\sf CSP}}

\newcommand{\N}{\mathbb{N}}

\nc{\BN}{\mathbb{N}}

\nc{\BZ}{\mathbb{Z}}

\newcommand{\ind}{\mathbf{1}}
\newcommand{\bA}{\mathbf{A}}
\nc{\tbA}{\widetilde{\bA}}

\newcommand{\cP}{\mathcal{P}}

\DeclareMathOperator{\argmax}{argmax}

\DeclareMathOperator{\argmin}{argmin}
\newcommand{\eps}{\varepsilon}

\newcommand{\val}{\mathrm{val}}
\newcommand{\bpsi}{\boldsymbol{\psi}}
\newcommand{\bPsi}{\boldsymbol{\Psi}}
\newcommand{\lrsg}{\leftrightsquigarrow}
\newcommand{\tPi}{\widetilde{\Pi}}
\newcommand{\tV}{\widetilde{V}}
\newcommand{\tG}{\widetilde{G}}
\newcommand{\tE}{\widetilde{E}}
\newcommand{\tC}{\widetilde{C}}
\newcommand{\tA}{\widetilde{A}}
\newcommand{\tSigma}{\widetilde{\Sigma}}
\newcommand{\tpsi}{\widetilde{\psi}}
\newcommand{\tS}{\widetilde{S}}
\newcommand{\ttt}{\widetilde{t}}
\newcommand{\tv}{\widetilde{v}}

\newcommand{\te}{\widetilde{e}}

\newcommand{\tsigma}{\widetilde{\sigma}}

\newcommand{\dec}{\mathsf{Dec}}
\newcommand{\asgnt}{\mathsf{AsgnT}}

\newcommand{\minlab}{\textsc{MinLab}}
\newcommand{\maxpar}{\textsc{MaxPar}}
\newcommand{\SAT}{\textsc{SAT}}

\newcommand{\gapcsp}[3]{({#1,#2})\text{-}\mathrm{Gap}\text{-}\mathrm{2\text{-}\CSP}_{#3}}
\newcommand{\gapminmaxcsp}[3]{({#1,#2})\text{-}\mathrm{Gap}\mathrm{MinMax} \text{-}2\text{-}\mathrm{\CSP}_{#3}}
\newcommand{\gapmaxmincsp}[3]{({#1,#2})\text{-}\mathrm{Gap}\mathrm{MaxMin} \text{-}k\text{-}\mathrm{\CSP}_{#3}}
\newcommand{\gapmaxminthreecsp}[3]{({#1,#2})\text{-}\mathrm{Gap}\mathrm{MaxMin} \text{-}3\text{-}\mathrm{\CSP}_{#3}}
\newcommand{\gapmaxmintwocsp}[3]{({#1,#2})\text{-}\mathrm{Gap}\mathrm{MaxMin} \text{-}2\text{-}\mathrm{\CSP}_{#3}}
\newcommand{\gapmaxminvarcsp}[4]{({#1,#2})\text{-}\mathrm{Gap}\mathrm{MaxMin} \text{-}{#3}\text{-}\mathrm{\CSP}_{#4}}
\newcommand{\enc}{\mathsf{Enc}}

\newtheorem{theorem}{Theorem}

\newtheorem{lemma}[theorem]{Lemma}
\newtheorem{claim}[theorem]{Claim}
\newtheorem{definition}[theorem]{Definition}

\newtheorem{corollary}[theorem]{Corollary}
\newtheorem{remark}[theorem]{Remark}

\title{On Inapproximability of Reconfiguration Problems:\\ \PSPACE-Hardness and some Tight \NP-Hardness Results\footnote{This paper merges the preprints \cite{KM24} and \cite{GRW25}.}}

\date{}

\author{
Venkatesan Guruswami\thanks{ 
Simons Institute for the Theory of Computing and UC Berkeley. 
Research supported in part by NSF grants CCF-2228287 and CCF-2211972 and a Simons Investigator award.} \\{\tt venkatg@berkeley.edu}
\and 
Karthik C.\ S.\thanks{
Rutgers University.
Research supported by NSF grant CCF-2313372 and by the Simons Foundation, Grant Number 825876, Awardee Thu D. Nguyen.}\\ \texttt{karthik.cs@rutgers.edu}
\and
Pasin Manurangsi\thanks{ 
Google Research.
Part of this work was done while the author was visiting Rutgers University.}\\\texttt{pasin@google.com}
\and
Xuandi Ren\thanks{
UC Berkeley. 
Research supported in part by NSF CCF-2228287.}\\{\tt xuandi\_ren@berkeley.edu}
\and
Kewen Wu\thanks{
Institute for Advanced Study.}\\{\tt shlw\_kevin@hotmail.com}
}

\begin{document}
\maketitle

\thispagestyle{empty}

\begin{abstract}
The field of combinatorial reconfiguration studies search problems with a focus on transforming one feasible solution into another. \vspace{0.1cm}

Recently, Ohsaka [STACS'23] put forth the \emph{Reconfiguration Inapproximability Hypothesis} (\RIH), which roughly asserts that there is some $\varepsilon>0$ such that given as input a $k$-\CSP\ instance (for some constant $k$) over some constant sized alphabet, and two satisfying assignments $\psi_s$ and $\psi_t$, it is \PSPACE-hard to find a sequence of assignments starting from  $\psi_s$ and ending at $\psi_t$ such that every assignment in the sequence satisfies at least $(1-\varepsilon)$ fraction of the constraints and also that every   assignment in the sequence is obtained by changing its immediately preceding assignment (in the sequence) on exactly one variable. 
Assuming \RIH, many important reconfiguration problems have been shown to be \PSPACE-hard to approximate by Ohsaka [STACS'23; SODA'24]. \vspace{0.1cm}

In this paper, we prove \RIH, and consequently obtain constant-factor 
\PSPACE-hardness of approximation results for many reconfiguration 
problems, thereby resolving an open question posed by Ito et al.\ 
[TCS'11]. Our proof uses known constructions of 
Probabilistically Checkable Proofs of Proximity (in a black-box manner) 
to create the gap, and further leverages a parallelization framework 
from recent parameterized inapproximability results to analyze the 
quantitative trade-off between $\varepsilon$ and $k$ in \RIH. 
We note that Hirahara and Ohsaka [STOC'24] have also independently proved \RIH.
\vspace{0.1cm}

We also prove that the aforementioned  $k$-\CSP\ Reconfiguration problem is \NP-hard to approximate to within a factor of $1/2 + \varepsilon$ (for any $\varepsilon>0$) when $k=2$. We complement this with a polynomial time $(1/2 - \varepsilon)$-approximation algorithm, which improves upon a $(1/4 - \varepsilon)$-approximation algorithm of Ohsaka [2023] (again for any $\varepsilon>0$).\vspace{0.1cm}

Finally, we show that Set Cover Reconfiguration is \NP-hard to approximate to within a factor of $2 - \varepsilon$ for any constant $\varepsilon > 0$, which matches the simple linear-time 2-approximation algorithm by Ito et al. [TCS'11]. 
\end{abstract}

\clearpage
\tableofcontents
\thispagestyle{empty}
\clearpage
\setcounter{page}{1}
\section{Introduction}

\emph{Combinatorial reconfiguration}  is a field of research that investigates the following question: Is it possible to find a step-by-step transformation between two feasible solutions for a search problem while preserving their feasibility?
The field has received a lot of attention over the last two decades with applications to various real-world scenarios that are either dynamic or uncertain (please see the surveys of Nishimura \cite{nishimura2018introduction} or
van den Heuvel \cite{heuvel13complexity}, or the thesis of Mouawad \cite{mouawad2015reconfiguration} for details).

Many of the  reconfiguration problems studied in literature are derived from classical search problems, for example, consider the classical Set Cover problem, where given a collection of subsets of a universe and an integer $k$, the goal is to find $k$ input subsets whose union is the universe (such a collection of $k$ subsets is referred to as a {set cover}). In the corresponding reconfiguration problem, referred to as the \emph{Set Cover Reconfiguration} problem \cite{hearn2005pspace,ito2011complexity,haddadan2016complexity}, we are given as input a collection of subsets of a universe, an integer $k$, and two set covers $T_{s}$ and $T_t$, both of size $k$, and the goal is to find a sequence of set covers  starting from $T_s$ and ending at $T_t$ such that every  set cover in the sequence has small size and is obtained from its immediately preceding set cover in the sequence by adding or removing exactly one subset (see \Cref{sec:prelimSetCover} for a formal definition). In similar spirit, various reconfiguration analogues of important search problems have been studied in literature, such as  
{Boolean Satisfiability} \cite{gopalan2009connectivity,makino2011exact,mouawad2017shortest},
{Clique} \cite{ito2011complexity},  {Vertex Cover} \cite{bonsma2016independent,bonsma2014reconfiguring,kaminski2012complexity,lokshtanov2019complexity,wrochn2018reconfiguration}, {Matching} \cite{ito2011complexity},
{Coloring} \cite{cereceda2008connectedness,bonsma2009finding,cereceda2011finding},
{Subset Sum} \cite{ito2014approximability}, and
{Shortest Path} \cite{kaminski2011shortest,bonsma2013complexity}.

Many of the above mentioned reconfiguration problems are \PSPACE-hard, and thus to address this intractability, Ito et al.~\cite{ito2011complexity} initiated the study of these problems under the lens of \emph{approximation}. For many of the above mentioned reconfiguration problems, polynomial time (non-trivial) approximation algorithms are now known \cite{ito2011complexity,MatsuokaO21,Ohsaka-label-cover}. On the other hand, using the \PCP\ theorem for \NP\  \cite{AS98,ALMSS98,Dinur07}, there are also some \NP-hardness of approximation results for reconfiguration problems, such as for the reconfiguration analogs of Boolean satisfiability \cite{ito2011complexity}, Clique \cite{ito2011complexity}, Binary arity Constraint Satisfaction Problems (hereafter 2-\CSP) \cite{Ohsaka23prev,ohsaka2023gap}, and Set Cover \cite{ohsaka2023gap}. 
However, a glaring gap in the above results is that we only know of \NP-hardness of approximation results for reconfiguration problems, but their exact versions are known to be \PSPACE-hard. Thus, the authors of \cite{ito2011complexity} posed the following question, 
\begin{center}
    ``\emph{Are the problems in Section 4 \PSPACE-hard to approximate (not just \NP-hard)?}'',
\end{center}
while referring to the problems in Section~4 of their paper, for which they had shown \NP-hardness of approximation results. 

To address the above question, Ohsaka in \cite{Ohsaka23prev}, introduced the \emph{Reconfiguration Inapproximability Hypothesis} (\RIH),  a reconfiguration analogue of the \PCP\ theorem for \NP, and assuming which he showed that it is \PSPACE-hard to approximate reconfiguration analogs of 2-\CSP, Boolean satisfiability,  Indepedent set, Clique, and Vertex Cover, to some constant factor bounded away from 1. These \PSPACE-hardness of approximation factors were later improved (still under \RIH) in \cite{ohsaka2023gap} for the reconfiguration analogue of 2-\CSP\ and the Set Cover Reconfiguration problem to 0.9942 and 1.0029 respectively. 

\begin{sloppypar}In order to state \RIH, we first define the reconfiguration analog of  gap-\CSP\ (which is the centerpiece problem of the hardness of approximation in \NP).  In the $\gapmaxmincsp{1}{1-\eps}{q}$ problem, we are given as input (i) a $k$-uniform hypergraph and every hyperedge of this hypergraph corresponds to a constraint on $k$ variables (which are the $k$ nodes of the hyperedge) over an alphabet set of size $q$, and (ii) two satisfying assignments $\psi_s$ and $\psi_t$ to the $k$-\CSP\ instance. For an instance of $\gapmaxmincsp{1}{1-\eps}{q}$, a reconfiguration assignment sequence is a sequence of  assignments starting from $\psi_s$ and ending at $\psi_t$ such that    every   assignment in the sequence is obtained by changing its immediately preceding assignment (in the sequence) on exactly one variable. The goal is then to distinguish the completeness case from the soundness case: in the completeness case, there exists  a reconfiguration assignment sequence such that every assignment in the sequence is a satisfying assignment, and in the soundness case, in every reconfiguration assignment sequence there is some assignment in the sequence which violates at least $\varepsilon$ fraction of the constraints 
 (see \Cref{sec:prelimCSP} for a formal definition of $\gapmaxmincsp{1}{1-\eps}{q}$). \RIH\ then asserts that there exists universal constants $\eps>0$, and $q,k\in\mathbb{N}$ such that $\gapmaxmincsp{1}{1-\eps}{q}$ is \PSPACE-hard\footnote{Note that the \PCP\ theorem is equivalent to the following statement (see \cite{Dinur07}): For some $\varepsilon>0$, it is \NP-hard to approximate 2-\CSP{}s to $(1-\varepsilon)$ factor.}.\end{sloppypar}

Given that \RIH\ implies the aforementioned \PSPACE-hardness of approximation results for many important problems, we naturally have the following fundamental question in the study of (in)approximability of reconfiguration problems: 

\begin{center}
    \emph{Is \RIH\ true?}
\end{center}

We remark that Ito et al.~\cite{ito2011complexity} showed that $\gapmaxmincsp{1}{1-\eps}{q}$ is \NP-hard for $q=2$, $k=3$, and some\footnote{In fact, they achieve $\eps = 1/16 - o(1)$.} $\eps > 0$. However, proving \RIH\ (i.e., proving \PSPACE-hardness of $\gapmaxmincsp{1}{1-\eps}{q}$) had remained elusive. Recently, \RIH\ has been resolved independently by Hirahara–Ohsaka (via the framework of probabilistically checkable reconfiguration proofs, PCRP) \cite{hirahara2023probabilistically} and by us (this paper). Moreover, Ohsaka’s subsequent works \cite{ohsaka2023gap,ohsaka2024alphabet} give a Dinur-style route, i.e., combining gap amplification and alphabet reduction, towards proving \RIH.


\subsection{Our Contribution}

Our primary contribution is proof of \RIH\ (see \Cref{sec:result1}). 
In addition, we establish a sharp correspondence between the soundness gap of \RIH\ and the rejection probability of constant-query assignment testers via a simple parallelization trick (see \Cref{sec:result1.2}).
Finally, we also provide tight results on the \NP-hardness of approximating GapMaxMin-2-\CSP$_q$ and Set Cover Reconfiguration (see \Cref{sec:result2}).

\subsubsection[PSPACE-Hardness of Approximation of Reconfiguration: Resolving RIH]{\PSPACE-Hardness of Approximation of Reconfiguration: Resolving \RIH}\label{sec:result1}

The main result of our work is providing a  proof of \RIH:
\begin{theorem} \label{thm:rih}
There exist constants $q\in\mathbb{N}, k\in\mathbb{N}, \eps > 0$ such that $\gapmaxmincsp{1}{1-\eps}{q}$ is \PSPACE-hard.
\end{theorem}

Our proof does \emph{not} follow Ohsaka's approach~\cite{ohsaka2023gap} of adapting Dinur's proof. Rather, we give a rather direct (and short) proof of \RIH\ by using efficient constructions of \emph{Probabilistically Checkable Proofs of Proximity} (\PCPP) in a black-box way (see \Cref{sec:proof-overview} for details).

As a consequence of Ohsaka's aforementioned reductions~\cite{Ohsaka23prev}, we immediately get the following as a corollary. (We do not define all the problems here and refer to~\cite{Ohsaka23prev} for more details.)

\begin{corollary}
There exists a constant $\eps > 0$, such that the following are all \PSPACE-hard:
\begin{itemize}
\item $(1 - \eps)$-approximation of 3SAT Reconfiguration,
\item $(1 + \eps)$-approximation of Vertex Cover Reconfiguration (and thus Set Cover Reconfiguration),
\item $(1 - \eps)$-approximation of Independent Set Reconfiguration and Clique Reconfiguration,
\item $\gapmaxmintwocsp{1}{1-\eps}{3}$.
\end{itemize}
\end{corollary}

This paper merges \cite{KM24} and \cite{GRW25}. The former appeared around the same time as the independent work of Hirahara and Ohsaka \cite{hirahara2023probabilistically}, which also proved \RIH\ using \PCPP\ but through a different approach that relies on additional properties of the \PCPP\ (see \Cref{rem:HO24} for more details).

{
\subsubsection{Sharp Query and Soundness Trade-offs of \RIH\ by Parallelization}\label{sec:result1.2}

The state-of-the-art set of parameters in the PSPACE-hardness of $\gapmaxmincsp{1}{1-\eps}{q}$ is $\eps=0.9942$, $k=2$, and $q=O(1)$ \cite{ohsaka2023gap}\footnote{A recent work by Ohsaka \cite{ohsaka2024alphabet} shows the \PSPACE-hardness of $\gapmaxmincsp{1}{1-\eps}{q}$ for $\eps>10^{-18}$, $k=2$, and $q<2\times10^6$. Furthermore, as we will later show, if the arity $k$ and the alphabet size $q$ are arbitrarily large constants, then the soundness can be arbitrarily small. Hence we focus on the soundness when $k=2$ and $q=O(1)$ here.}.
Note that the $0.0058$ soundness gap will become even worse in downstream applications \cite{Ohsaka23prev}; and it is unclear if the gap amplification techniques in \cite{ohsaka2023gap} can be improved to give significantly better bounds.

To address this issue, we apply a parallelization trick, inspired by recent works in parameterized hardness of approximation \cite{LRSW23,GLRSW24,guruswami2024almost,BafnaKM25}, to obtain a sharp connection between the soundness gap of \RIH\ and the one of assignment tester (a.k.a \PCPP).
We introduce assignment tester formally in Section~\ref{sec:pcpp}. At this point, let us just comment that it is a widely studied concept in complexity theory \cite{BGHSV06,DinurR06,Harsha2004Robust}.

\begin{theorem}\label{thm:main_grw}
For some fixed constant $\kappa>0$, if for every Boolean circuit of size $n$, there is a polynomial-time $k$-query assignment tester with proximity parameter $\kappa$ and rejection probability $\gamma$, then for some alphabet of constant size $q$, $\gapmaxminvarcsp{1}{1-\gamma}{(k+1)}{q}$ is \PSPACE-hard.
\end{theorem}

In other words, up to an extra query (arity of \CSP), the parameters of \RIH\ and assignment tester are identical.
In addition, the idea of parallelization is fairly general, and it works as long as there are multiple testers with the same structure of queries used by the verifier.

To get a result for arity $k=2$, we note that \cite[Lemma 5.4]{ohsaka2024alphabet} gives a way to trade the arity of the CSP with the soundness gap. Combined with \Cref{thm:main_grw}, this leads to the following corollary. 

\begin{corollary}\label{cor:main_grw}
For some fixed constant $\kappa>0$, if for every Boolean circuit of size $n$, there is a polynomial-time $k$-query assignment tester with proximity parameter $\kappa$ and rejection probability $\gamma$, then for some alphabet of constant size $q$, $\gapmaxmintwocsp{1}{1-\frac\gamma{k+1}}{q}$ is \PSPACE-hard.
\end{corollary}

Due to the lack of assignment tester constructions with explicit constant query complexity and soundness, we are not yet able to use \Cref{cor:main_grw} to improve the $0.9942$ soundness in \cite{ohsaka2023gap}. 
We view our \Cref{thm:main_grw} and \Cref{cor:main_grw} as another concrete motivation to understand the soundness-query tradeoff of assignment testers.

Finally, we briefly explain the testers needed in \Cref{thm:main_grw} and defer its formal definition to \Cref{sec:pcpp}. 
Let $\Phi\colon\{0,1\}^n\to\{0,1\}$ be a Boolean function, given by its circuit representation.
Intuitively, the tester is a verifier $\mathcal V$ that efficiently decides, given access to an auxiliary proof $\pi$, if a particular input $x\in\{0,1\}^n$ is a solution of $\Phi$ or $x$ is far from any solution of $\Phi$.
The \emph{query complexity} is the maximum number of bits that $\mathcal V$ reads on $(x,\pi)$; the \emph{proximity parameter} is the distance of $x$ to solutions; the \emph{rejection probability} is the verifier's rejection probability given $x$ is far from solutions.
The assignment tester is implicit in the low-degree testers \cite{AS98,ALMSS98} that utilize auxiliary oracles to test if a given function is close to a low-degree polynomial. It is then explicitly defined in \cite{BGHSV06}, and concurrently in \cite{DinurR06} with the name \emph{Assignment Tester}. We refer interested readers to Harsha's thesis \cite{Harsha2004Robust} for an overview of the history.

While the PCP theorem admits verifiers of subconstant soundness (i.e., rejection probability $1-o(1)$) and two queries~\cite{moshkovitz2008two}, we are not aware of explicit tradeoffs between soundness and query complexity for assignment testers.
}

\subsubsection[Tight NP-Hardness of Approximation for GapMaxMin-2-CSP and Set Cover Reconfiguration]{Tight \NP-Hardness of Approximation for GapMaxMin-2-\CSP$_q$ and Set Cover Reconfiguration}\label{sec:result2}

In addition to the previously mentioned \PSPACE-hardness results, we also provide \NP-hardness results that have improved (and nearly tight) inapproximability ratios.

In \cite{ohsaka2023gap}, the author showed that for every $\varepsilon>0$ there exists  some $q\in\mathbb N$ such that deciding $\gapmaxmintwocsp{1}{3/4+\eps}{q}$ is \NP-hard. We improve this result to the following.
\begin{theorem} \label{thm:maxmin-csp-np-hardness}
For any $\eps > 0$ there exists $q \in \N$ such that 
deciding $\gapmaxmintwocsp{1}{1/2+\eps}{q}$ is \NP-hard.
\end{theorem}

The above hardness result is the essentially the best possible result, because in \Cref{sec:ourcontrialg}, we will show that for every $\varepsilon>0$ and every $q\in\mathbb N$,  deciding $\gapmaxmintwocsp{1}{1/2-\eps}{q}$ is in \P.

Next, we consider the (in)approximability of the Set Cover Reconfiguration problem.  Previously, in \cite{ohsaka2023gap}, the author showed that  approximating the Set Cover Reconfiguration problem to 1.0029 factor is \NP-hard. We improve this result to the following.
\begin{theorem} \label{thm:minmax-setcover-np-hardness}
For any $\eps > 0$, it is \NP-hard to approximate Set Cover Reconfiguration to within a factor of $(2 - \eps)$.
\end{theorem}

The above hardness result is the essentially the best possible result, because in \cite{ito2011complexity}, the authors show that approximating Set Cover Reconfiguration to a factor of 2 is in \P.

As an intermediate step in the above result for the Set Cover Reconfiguration problem, we prove the tight inapproximability of the minimization variant of the GapMaxMin-2-\CSP$_q$ problem, namely, the  $\gapminmaxcsp{1}{s}{q}$ problem, where  we are given the same  input as for the GapMaxMin-2-\CSP$_q$ problem,  and the goal is then to distinguish the completeness case from the soundness case: in the completeness case, there exists  a reconfiguration assignment sequence such that every assignment in the sequence is a satisfying assignment, and in the soundness case, in every reconfiguration satisfying multiassignment sequence\footnote{A reconfiguration satisfying multiassignment sequence is the same as a reconfiguration assignment sequence but where we allow each sequence element to be a multiassignment instead of just an assignment and also insist that every multiassignment in the sequence satisfies all the constraints.} there is some multiassignment in the sequence whose average number of labels per variable is more than $s$. 
 (See \Cref{sec:prelimCSP} for a formal definition of $\gapminmaxcsp{1}{s}{q}$).

\begin{theorem} \label{thm:minmax-csp-np-hardness}
For any $\eps > 0$ there exists $q \in \N$ such that 
deciding $\gapminmaxcsp{1}{2-\eps}{q}$ is \NP-hard.
\end{theorem}

Again, we note that the above hardness result is the best possible, as there is a 2-factor polynomial time approximation algorithm for the GapMinMax-2-\CSP$_q$ problem via the same approach as Ito et al.'s algorithm for Set Cover Reconfiguraiton~\cite{ito2011complexity}. In particular, we first start from $\psi_s$ and sequentially include the assignment in $\psi_t$ to each of the variables, to eventually obtain the multiassignment $\tilde \psi$ where for each variable we have both the assignments of $\psi_s$ and $\psi_t$ to it as part of the multiassignment. Then starting from $\tilde \psi$ we sequentially remove the assignment to each variable in $\psi_s$ to eventually obtain $\psi_t$.

\subsubsection[Improved Approximation Algorithm for GapMaxMin-2-CSP]{Improved Approximation Algorithm for GapMaxMin-2-\CSP$_q$}\label{sec:ourcontrialg} 
For GapMaxMin-2-\CSP$_q$, the previous best polynomial-time approximation algorithm was due to Ohsaka~\cite{Ohsaka-label-cover}, which yields gives $(1/4 - \eps)$-approximation but only works when the average degree of the graph is sufficiently large. We improve on this, by giving a $(1/2 - \eps)$-approximation algorithm which works even \emph{without} the average degree assumption:

\begin{theorem} \label{thm:approx-algo-2csp}
For any constant $\eps \in (0, 1/2]$ and $q \in \N$, there exists a (randomized) polynomial-time  algorithm 
for $\gapmaxmintwocsp{1}{1/2 - \eps}{q}$.
\end{theorem}

As mentioned earlier, this matches our \NP-hardness provided in \Cref{thm:maxmin-csp-np-hardness}. We also note that our algorithm running time is in fact completely independent of the alphabet size $q$ (and thus can handle arbitrarily large $q$). 

\subsection{Proof Overview}
\label{sec:proof-overview}
In this subsection, we provide the proof overview for all of the main results in this paper. 

\subsubsection[Proof of RIH and its Query-Soundness Trade-offs]{Proof of \RIH\ and its Query-Soundness Trade-offs}

In order to prove \Cref{thm:rih}, we need two tools. The first is just a \emph{good binary code} (see \Cref{sec:code} for relevant definitions). We use the notation $\enc: \{0,1\}^k \to \{0,1\}^n$ to denote an encoding algorithm for a code  of message length $k$  and block length $n$ whose distance is $d_{\enc}$. We use a code that has positive constant rate and relative distance.

The second tool we need is an \emph{assignment tester} (a.k.a.\ Probabilistically Checkable Proofs of Proximity) which is simply an algorithm   which takes as input   a Boolean circuit $\Phi$ and outputs a 2-\CSP\ whose variable set  is a superset  of the variable set of $\Phi$ with the following guarantee. For every assignment $\psi$ of $\Phi$, and any extension of $\psi$ to a total assignment to the variables of the 2-\CSP, we have that the number of unsatisfied constraints in the 2-\CSP\ w.r.t.\ that assignment is linearly proportional to the distance to the closest  assignment to $\psi$ (under Hamming distance) which makes $\Phi$ evaluate to 1 
 (see \Cref{sec:pcpp} for a formal definition).

It is known that $\gapmaxmintwocsp{1}{1}{q}$ is \PSPACE-hard even for $q=3$ \cite{GopalanKMP09}. 
Given a $\gapmaxmintwocsp{1}{1}{q}$ instance $\Pi = (G = (V, E), \Sigma, \{C_e\}_{e \in E})$ with two satisfying assignments $\psi_s,\psi_t:V\to \Sigma$, we construct  a $\gapmaxminthreecsp{1}{1-\eps}{q_0}$ instance $\tPi = (\tG = (\tV, \tE), \tSigma, \{\tC_e\}_{e \in \tE})$ where, $$\tV := \{v^*\} \uplus \left(\biguplus_{i \in [4]} \tV_i\right) \uplus \left(\biguplus_{i \in [4]} \tA_i\right) \quad\text{ and }\quad \tE := \biguplus_{i \in [4]} \tE_i$$ where the variables and constraints are defined as follows. 
\begin{description}
\item[Vertex Set:] For every $i \in [4]$, let $\tV_i$ denote a set of $n$ fresh variables, where $n$ is the block length of the code given by $\enc$ whose message length  is $|V|\cdot \lceil\log q\rceil $. We use the notation that for  every $i\in[4]$, let $\bar{i}_1, \bar{i}_2,$ and $\bar{i}_3$ denote the elements of $[4] \setminus \{i\}$ and define  the Boolean circuit $\Phi_i$   on variable set $\tV_{\bar{i}_1}\uplus \tV_{\bar{i}_2}\uplus \tV_{\bar{i}_3}$   which   evaluates   to 1 if and only if  for all $\ell,\ell'\in [3]$ 
   we have the input to $\tV_{\bar{i}_{\ell}}$ is a valid codeword such that it's decoded message is a satisfying assignment to the variables of $\Pi$, and   the decoded input to $\tV_{\bar{i}_{\ell}}$ and the decoded input to $\tV_{\bar{i}_{\ell'}}$ differ by at most 1 in Hamming distance.  
  Let $\Pi'_i = (G'_i = (V'_i, E'_i), \tilde\Sigma, \{C^i_e\}_{e \in E'})$ be the 2-\CSP$_{q_0}$ instance produced by applying the assignment tester on $\Phi_i$ where $V'_i = \tV_{\bar{i}_1} \uplus \tV_{\bar{i}_2} \uplus \tV_{\bar{i}_3} \uplus \tA_i$ (i.e., for each $i\in[4]$, $\tA_i$ is the additional set of variables produced by the assignment tester). 
 \item[Hyperedge Set and Constraints:] For all $i\in[4]$, and for each $e = (u, v) \in E'_i$, create a hyperedge $\te = (v^*, u, v)$ in $\tE_i$ with the following constraint:
\begin{align*}
\forall \tsigma^*, \tsigma_u, \tsigma_v\in\Sigma,\ C_{\te}(\tsigma^*, \tsigma_u, \tsigma_v) = 1 \Longleftrightarrow  \left(\left(\left(\tsigma^* = i\right) \wedge \left(C^i_{e}(\tsigma_u, \tsigma_v)=1\right)\right) \text{ or }\left(\tsigma^* \in [4] \setminus \{i\}\right)\right).
\end{align*}
See \Cref{fig} for an illustration of the design of these constraints. 
\item[Beginning and End of the Reconfiguration Assignment Sequence:] In order to define $\tpsi_s$ and $\tpsi_t$,  the starting and terminating satisfying assignments of $\tPi$, it will be convenient to first define an additional notion. For every satisfying assignment $\psi: V \to \Sigma$ of $\Pi$, we define an assignment $\psi^{\enc}: \tV \to \tilde\Sigma$ of $\tPi$ in the following way. First, fix $i\in[4]$ and we know that if the input to $\tV_{\bar{i}_1},\tV_{\bar{i}_2},$ and $\tV_{\bar{i}_3}$ are all equal to $\enc(\psi)$ then   $\Phi_i$  evaluates to 1 and thus from the property of the assignment tester, there is some assignment to the variables in $\tA_i$, say $\phi_i$, such that all constraints of $\Pi'_i$ are satisfied.  

Then, we define $\psi^{\enc}$ as follows:
\begin{align*}
\psi^{\enc}(v^*) &= 4, \\
\psi^{\enc}|_{\tV_i} &= \enc(\psi) &\forall i \in [4], \\
\psi^{\enc}|_{\tA_i} &= \phi_{i} &\forall i \in [4].
\end{align*}
 Then, let $\tpsi_s = (\psi_s)^{\enc}$ and $\tpsi_t = (\psi_t)^{\enc}$ and it can be verified that both are satisfying assignments to $\tPi$. 
\end{description}

This construction is inspired by similar ideas appearing in literature in papers concerning hardness and lower bounds of fixed point computation (for example see \cite[{Section 4.1}]{rubinstein2018inapproximability}).

\begin{figure}[!ht]
    \centering
\resizebox{\textwidth}{!}{
\begin{tikzpicture}\huge
\filldraw [fill=pink!50!white,thick] plot [smooth cycle] coordinates {(1.25,0.15) (0.75,0.9) (-0.9,0.7) (-1.25,0.15) (-0.75,-0.9) (0.9,-0.7)};
\filldraw (0,0) circle (3pt);
\filldraw (0.5,0.7) circle (3pt);
\filldraw (-0.5,-0.7) circle (3pt);
\filldraw (1,0) circle (3pt);
\filldraw (-1,0) circle (3pt);
\filldraw (0.7,-0.5) circle (3pt);
\filldraw (-0.7,0.5) circle (3pt);
\node at (1.5,1)  {$\tA_3$}; 

\filldraw [fill=pink!50!white,thick,rotate=90,yshift = -400] plot [smooth cycle] coordinates {(1.25,0.15) (0.75,0.9) (-0.9,0.7) (-1.25,0.15) (-0.75,-0.9) (0.9,-0.7)};
\filldraw[rotate=90,yshift = -400] (0,0) circle (3pt);
\filldraw[rotate=90,yshift = -400] (0.5,0.7) circle (3pt);
\filldraw[rotate=90,yshift = -400] (-0.5,-0.7) circle (3pt);
\filldraw[rotate=90,yshift = -400] (1,0) circle (3pt);
\filldraw[rotate=90,yshift = -400] (-1,0) circle (3pt);
\filldraw[rotate=90,yshift = -400] (0.7,-0.5) circle (3pt);
\filldraw[rotate=90,yshift = -400] (-0.7,0.5) circle (3pt);
\node[yshift = 340] at (1.35,1)  {$\tA_1$}; 

\filldraw [fill=pink!50!white,thick,rotate=-90,xshift = -400] plot [smooth cycle] coordinates {(1.25,0.15) (0.75,0.9) (-0.9,0.7) (-1.25,0.15) (-0.75,-0.9) (0.9,-0.7)};
\filldraw[rotate=-90,xshift = -400] (0,0) circle (3pt);
\filldraw[rotate=-90,xshift = -400] (0.5,0.7) circle (3pt);
\filldraw[rotate=-90,xshift = -400] (-0.5,-0.7) circle (3pt);
\filldraw[rotate=-90,xshift = -400] (1,0) circle (3pt);
\filldraw[rotate=-90,xshift = -400] (-1,0) circle (3pt);
\filldraw[rotate=-90,xshift = -400] (0.7,-0.5) circle (3pt);
\filldraw[rotate=-90,xshift = -400] (-0.7,0.5) circle (3pt);
\node[xshift = 320] at (1.35,1)  {$\tA_4$};

\filldraw [fill=pink!50!white,thick,rotate=180,xshift = -400,yshift = -400] plot [smooth cycle] coordinates {(1.25,0.15) (0.75,0.9) (-0.9,0.7) (-1.25,0.15) (-0.75,-0.9) (0.9,-0.7)};
\filldraw[rotate=180,xshift = -400,yshift = -400] (0,0) circle (3pt);
\filldraw[rotate=180,xshift = -400,yshift = -400] (0.5,0.7) circle (3pt);
\filldraw[rotate=180,xshift = -400,yshift = -400] (-0.5,-0.7) circle (3pt);
\filldraw[rotate=180,xshift = -400,yshift = -400] (1,0) circle (3pt);
\filldraw[rotate=180,xshift = -400,yshift = -400] (-1,0) circle (3pt);
\filldraw[rotate=180,xshift = -400,yshift = -400] (0.7,-0.5) circle (3pt);
\filldraw[rotate=180,xshift = -400,yshift = -400] (-0.7,0.5) circle (3pt);
\node[,yshift=420,xshift = 320] at (1.25,-1.5)  {$\tA_2$};

\filldraw [xshift = 120,yshift = 120,fill=blue!30!white,thick] plot [smooth cycle] coordinates {(0.15,1.25) (0.75,0.9) (0.9,-0.7) (0.15,-1.25) (-0.75,-0.9) (-0.9,0.7)};
\filldraw[xshift = 120,yshift = 120] (0,0) circle (3pt);
\filldraw[xshift = 120,yshift = 120] (0.5,0.7) circle (3pt);
\filldraw[xshift = 120,yshift = 120] (-0.5,-0.7) circle (3pt);
\filldraw[xshift = 120,yshift = 120] (0,-1) circle (3pt);
\filldraw[xshift = 120,yshift = 120] (0,1) circle (3pt);
\filldraw[xshift = 120,yshift = 120] (0,-0.5) circle (3pt);
\filldraw[xshift = 120,yshift = 120] (0,0.5) circle (3pt);
\filldraw[xshift = 120,yshift = 120] (0.5,0) circle (3pt);
\filldraw[xshift = 120,yshift = 120] (-0.5,0) circle (3pt);
\filldraw[xshift = 120,yshift = 120] (-0.7,0.5) circle (3pt);
\filldraw[xshift = 120,yshift = 120] (0.7,-0.5) circle (3pt);
\node[yshift=120,xshift = 120] at (0,-1.75)  {$\tV_3$};

\filldraw [xshift = 280,yshift = 120,rotate=80,fill=blue!30!white,thick] plot [smooth cycle] coordinates {(0.15,1.25) (0.75,0.9) (0.9,-0.7) (0.15,-1.25) (-0.75,-0.9) (-0.9,0.7)};
\filldraw[xshift = 280,yshift = 120,rotate=80] (0,0) circle (3pt);
\filldraw[xshift = 280,yshift = 120,rotate=80] (0.5,0.7) circle (3pt);
\filldraw[xshift = 280,yshift = 120,rotate=80] (-0.5,-0.7) circle (3pt);
\filldraw[xshift = 280,yshift = 120,rotate=80] (0,-1) circle (3pt);
\filldraw[xshift = 280,yshift = 120,rotate=80] (0,1) circle (3pt);
\filldraw[xshift = 280,yshift = 120,rotate=80] (0,-0.5) circle (3pt);
\filldraw[xshift = 280,yshift = 120,rotate=80] (0,0.5) circle (3pt);
\filldraw[xshift = 280,yshift = 120,rotate=80] (0.5,0) circle (3pt);
\filldraw[xshift = 280,yshift = 120,rotate=80] (-0.5,0) circle (3pt);
\filldraw[xshift = 280,yshift = 120,rotate=80] (-0.7,0.5) circle (3pt);
\filldraw[xshift = 280,yshift = 120,rotate=80] (0.7,-0.5) circle (3pt);
\node[yshift=280,xshift = 120] at (0,1.5)  {$\tV_1$};

\filldraw [xshift = 120,yshift = 280,rotate=280,fill=blue!30!white,thick] plot [smooth cycle] coordinates {(0.15,1.25) (0.75,0.9) (0.9,-0.7) (0.15,-1.25) (-0.75,-0.9) (-0.9,0.7)};
\filldraw[xshift = 120,yshift = 280,rotate=280] (0,0) circle (3pt);
\filldraw[xshift = 120,yshift = 280,rotate=280] (0.5,0.7) circle (3pt);
\filldraw[xshift = 120,yshift = 280,rotate=280] (-0.5,-0.7) circle (3pt);
\filldraw[xshift = 120,yshift = 280,rotate=280] (0,-1) circle (3pt);
\filldraw[xshift = 120,yshift = 280,rotate=280] (0,1) circle (3pt);
\filldraw[xshift = 120,yshift = 280,rotate=280] (0,-0.5) circle (3pt);
\filldraw[xshift = 120,yshift = 280,rotate=280] (0,0.5) circle (3pt);
\filldraw[xshift = 120,yshift = 280,rotate=280] (0.5,0) circle (3pt);
\filldraw[xshift = 120,yshift = 280,rotate=280] (-0.5,0) circle (3pt);
\filldraw[xshift = 120,yshift = 280,rotate=280] (-0.7,0.5) circle (3pt);
\filldraw[xshift = 120,yshift = 280,rotate=280] (0.7,-0.5) circle (3pt);
\node[,yshift=120,xshift = 280] at (0,-1.5)  {$\tV_4$};

\filldraw [xshift = 280,yshift = 280,rotate=2,fill=blue!30!white,thick] plot [smooth cycle] coordinates {(0.15,1.25) (0.75,0.9) (0.9,-0.7) (0.15,-1.25) (-0.75,-0.9) (-0.9,0.7)};
\filldraw[xshift = 280,yshift = 280,rotate=2] (0,0) circle (3pt);
\filldraw[xshift = 280,yshift = 280,rotate=2] (0.5,0.7) circle (3pt);
\filldraw[xshift = 280,yshift = 280,rotate=2] (-0.5,-0.7) circle (3pt);
\filldraw[xshift = 280,yshift = 280,rotate=2] (0,-1) circle (3pt);
\filldraw[xshift = 280,yshift = 280,rotate=2] (0,1) circle (3pt);
\filldraw[xshift = 280,yshift = 280,rotate=2] (0,-0.5) circle (3pt);
\filldraw[xshift = 280,yshift = 280,rotate=2] (0,0.5) circle (3pt);
\filldraw[xshift = 280,yshift = 280,rotate=2] (0.5,0) circle (3pt);
\filldraw[xshift = 280,yshift = 280,rotate=2] (-0.5,0) circle (3pt);
\filldraw[xshift = 280,yshift = 280,rotate=2] (-0.7,0.5) circle (3pt);
\filldraw[xshift = 280,yshift = 280,rotate=2] (0.7,-0.5) circle (3pt);
\node[,yshift=280,xshift = 280] at (0,1.75)  {$\tV_2$};

\filldraw[xshift = 200,yshift = 200,red] (0,0) circle (3pt);
\node[,yshift=200,xshift = 200] at (0.2,-0.3)  {$v^*$};

\draw[line width=3mm,draw=blue!50!black,opacity=0.3] plot [smooth cycle] coordinates {(-1.25,-1.25) (-1.35,1.4) (2.5,10.7) (5.85,10.5) (5.6,5)   (10.75,5.85) (10.51,2.5)(1.25,-1.15)};

\draw[line width=3mm,draw=brown,rotate=90,yshift=-400,opacity=0.3] plot [smooth cycle] coordinates {(-1.25,-1.25) (-1.35,1.4) (2.5,10.7) (5.85,10.5) (5.6,5)   (10.75,5.85) (10.51,2.5)(1.25,-1.15)};

\draw[line width=3mm,draw=green!50!black,rotate=-90,xshift=-400,opacity=0.3] plot [smooth cycle] coordinates {(-1.25,-1.25) (-1.35,1.4) (2.5,10.7) (5.85,10.5) (5.6,5)   (10.75,5.85) (10.51,2.5)(1.25,-1.15)};

\draw[line width=3mm,draw=orange,rotate=180,yshift=-400,xshift=-400,opacity=0.3] plot [smooth cycle] coordinates {(-1.25,-1.25) (-1.35,1.4) (2.5,10.7) (5.85,10.5) (5.6,5)   (10.75,5.85) (10.51,2.5)(1.25,-1.15)};

\draw[line width=1mm,->] (-1,3) -- (-2,3);
\node at (-3.8,3.5) {Constraints};
\node at (-3.8,2.5) {of $E'_3$};

\draw[line width=1mm,->] (14.8,3) -- (15.8,3);
\node at (17.65,3.5) {Constraints};
\node at (17.65,2.5) {of $E'_4$};

\draw[line width=1mm,->] (-1,11.8) -- (-2,11.8);
\node at (-3.8,12.3) {Constraints};
\node at (-3.8,11.3) {of $E'_1$};

\draw[line width=1mm,->] (14.8,11.8) -- (15.8,11.8);
\node at (17.65,12.3) {Constraints};
\node at (17.65,11.3) {of $E'_2$};

\end{tikzpicture}}
    \caption{Reducing   $\gapmaxmintwocsp{1}{1}{q}$ instance $\Pi = (G = (V, E), \Sigma, \{C_e\}_{e \in E})$   to a $\gapmaxminthreecsp{1}{1-\eps}{q_0}$ instance $\tPi = (\tG = (\tV, \tE), \tSigma, \{\tC_e\}_{e \in \tE})$.}
    \label{fig}
\end{figure}

To show the completeness case, we prove that for any two satisfying assignment $\psi,\psi'$ of $\Pi$ we can construct a reconfiguration sequence starting at $\psi^{\enc}$ and ending at $(\psi')^{\enc}$ whose value is 1. 

In order to prove the soundness case, given a reconfiguration sequence for $\tPi$ of value $1-\varepsilon$, we construct a reconfiguration sequence for $\Pi$, essentially by a majority decoding argument. Note that every assignment to $\tPi$ has 4 potential assignments (not all distinct) to each variable in $V$ (i.e., $\tV_1,\ldots ,\tV_4$ each have an assignment to $V$). For every fixing of $i\in [4]$, when we look at a typical constraint in $E'_i$, then for each $v\in V$, we prove that there must be a clear majority assignment for $v$ given by the assignments to $\tV_{\bar{i}_1},\tV_{\bar{i}_2},$ and $\tV_{\bar{i}_3}$ (this is also why we needed \emph{four} copies of $\tV_i$s and not fewer). This along with few other claims finishes the analysis.

{
\paragraph{Parallelization.}
Now we discuss our tweak on top of the above analysis to prove \Cref{thm:main_grw}. For convenience, we will use the proof-verifier perspective of the assignment tester (see \Cref{rmk:pcpp_csp_game_grw} for details).

First of all, we note that there is no barrier in replacing the $2\text{-}\CSP_{q_0}$ instance (via a $2$-query assignment tester) with a $k\text{-}\CSP_{q_0}$ instance (via a $k$-query assignment tester).
The only change to the above proof is the final instance $\tilde\Pi$ has arity $k+1$ instead of $3=2+1$. For generality, we work with the general query complexity $k$.

The key observation is that the four $k\text{-}\CSP_{q_0}$ instances $\Pi_1',\Pi_2',\Pi_3',\Pi_4'$ in the construction have the same structure. In other words, they are the same constraints, albeit applied on different variables.
Specifically, for each randomness $r$ of the verifier, there is a subset of indices $I_r$, such that for each $i \in [4]$, the verifier $\Pi_i'$ queries the subset $I_r$ on the assignment of $V_i'=\tilde V_{\bar i_1}\uplus\tilde V_{\bar i_2}\uplus\tilde V_{\bar i_3}\uplus\tilde A_i$.

Given this observation, we ``parallelize'' the four instances in a natural way. 
We stack $\tilde V_1,\tilde V_2,\tilde V_3,\tilde V_4$ (and $\tilde A_1,\tilde A_2,\tilde A_3,\tilde A_4$) into four layers as $\tilde V$ (and $\tilde A$), with a quartic blowup on the alphabet size.
By querying a subset $I$ of assignments in $\tilde V,\tilde A$, we get the same subset of assignments in $\tilde V_j,\tilde A_j$ for every $j\in [4]$. 
After that, we perform the test in $\Pi_i'$, which involves $\tilde V_{\bar i_1},\tilde V_{\bar i_2},\tilde V_{\bar i_3},\tilde A_i$. 
Suppose our assignment tester has soundness $1-\gamma$ (i.e., rejection probability $\gamma$), applying this trick results in the PSPACE-hardness of $\gapmaxminvarcsp{1}{1-\gamma}{(k+1)}{q^4}$ over the alphabet $(\tilde\Sigma)^4$, where the one extra arity comes from $v^*$ and the quartic blowup comes from parallelizing four verifiers.

This idea of parallelization is inspired by recent works in parameterized hardness of approximation \cite{LRSW23,GLRSW24,guruswami2024almost}. In the world of parameterized complexity, we often encounter problems involving few variables but large alphabet size. 
The large alphabet size prevents direct applications of standard tools from the NP world, which typically requires a constant-sized alphabet. However, sometimes the problem is structured in a way that we can associate $[n]$ as $\Sigma^{\log _{|\Sigma|} n}$ for a constant-sized $\Sigma$, and view the constraints as acting in parallel across the $\log_{|\Sigma|} n$ layers of the variables. This allows us to apply known techniques, such as error correcting codes and PCPPs, in parallel to each layer, and verify the constraints across all layers simultaneously.}

\begin{remark}\label{rem:HO24}
    In \cite{hirahara2023probabilistically},  Hirahara and Ohsaka also prove \RIH, but their starting point is the Succinct Graph Reachability problem (that  is shown to be \PSPACE-complete), where given as input an efficiently computable  circuit $S:\{0,1\}^n\to \{0,1\}^n$ which is interpreted as providing the neighbor of a graph whose vertex set is $\{0,1\}^n$, the goal is to determine if vertex $\vec{1}$ can be reached from vertex $\vec{0}$. They interpret \RIH, as designing a probabilistic verifier for \PSPACE\ who makes constant queries, and then using a non-trivial combination of efficient \PCPP\ and \emph{Locally Testable Codes} (\LTC), they construct  such a verifier for the Succinct Graph Reachability problem. 
    Although our reduction and that in \cite{hirahara2023probabilistically} use a similar set of tools, there are a number of differences between the two. Chief among them is perhaps how the variables ``in transition'' are handled. Hirahara and Ohsaka create an additional symbol $\perp$ to denote the ``in transition'' state which requires modifications to the verifier and  subtle changes in the analysis; e.g. they need the \PCPP\ to be smooth~\cite{Par21}.
    Meanwhile, as explained above, we handle the ``in transition'' state by having multiple copies of the variables and an extra variable ($v^*$) to denote which copy is currently ``in transition''; this allows us to use any standard \PCPP\ without extra property.
\end{remark}

\subsubsection{Proof Overview of the Other Results}

\paragraph{Approximation Algorithm for GapMaxMin-2-\CSP$_q$\ (\Cref{thm:approx-algo-2csp}).} 
The idea of the approximation algorithm here is quite simple: We are going to change from the initial assignment $\psi_s$ to the final assignment $\psi_t$ by directly   changing the assignment sequentially to each of the  variables, one at a time. In other words, we define 
\begin{align*}
\psi_i(v) = 
\begin{cases}
\psi_s(v) &\text{ if } v \in S_i, \\
\psi_t(v) &\text{ otherwise,}
\end{cases}
\end{align*}
where $V = S_0 \subsetneq S_1 \subsetneq \cdots \subsetneq S_n = \emptyset$.
Notice that we never violate the edges inside $S_i$ or inside $V \setminus S_i$ at all. Thus, we may only violate the edges across the cut, i.e., with one endpoint in  $S_i$ and the other endpoint not in $S_i$, which we denote by $E[S_i, V \setminus S_i]$. Our main structural result is the following:

\begin{theorem}[Informal version of \Cref{thm:sequence-balanced-full}]
For any graph $G = (V, E)$ on $m$ edges, there exists an efficiently computable  downward sequence $V = S_0 \supsetneq \dots \supsetneq S_n = \emptyset$ such that for all $i\in[n]$, the number of edges between $S_i$ and $V\setminus S_i$ is at most  $ m\cdot \left(\frac{1}{2}+o(1)\right)$. 
\end{theorem}

The above result immediately yields the desired approximation guarantee.  
The overall strategy to prove the above theorem is in fact the same as that of~\cite{Ohsaka-label-cover}; the difference is that we derive a stronger bound $(1/2 + o(1))|E|$ instead of $(1/4 + o(1))|E|$ as in that work. To gain the intuition to our improved bound, notice that, if we pick each $S_i$ at random, at most $\frac{1}{2}|E|$ belong to $E[S_i, V \setminus S_i]$ \emph{in expectation for each $i \in [n]$}. Now, if we were able to achieve a high probability statement (with perhaps a slightly weaker bound), then we could try to use the union bound over all $i \in [n]$ to derive our desired lemma. This is roughly our strategy when the max-degree of the graph is small. However, if some vertices in the graph have large degrees, the standard deviation of $|E[S_i, V \setminus S_i]|$ is so large that one cannot hope for a high probability bound. This brings us to our final approach: we use a simple probabilistic argument on just the low-degree vertices to obtain the sequence of sets. Then, we use these low-degree vertices to ``vote on'' when to remove each high-degree vertex, i.e. removing it from $S_i$ only when at most half of its low-degree neighbor belongs to $S_i$.

\paragraph{\NP-Hardness for GapMaxMin-2-\CSP$_q$.}
Again, our hardness reduction approach is not too different from the previous work of Ohsaka~\cite{ohsaka2023gap}. Namely, we start from a Gap-2-\CSP$_q$\  instance $\tPi = (\tG = (\tV, \tE), \tSigma, \{\tC_e\}_{e \in \tE})$ that is \NP-hard to approximate (with very large inapproximability gap). Then, we introduce two additional special characters $(\sigma^*, 0)$ and $(\sigma^*, 1)$ to the alphabet set; each constraint is then extended such that, if the two special characters appear together, then it is unsatisfied. Otherwise, if only one of the two occurs, then it is satisfied. The starting assignment is then set to every variable being assigned $(\sigma^*, 0)$, while the ending assignment is then set such that every variable is assigned $(\sigma^*, 1)$.

There is a clear barrier to obtaining a $(1/2 + o(1))$-factor hardness of approximation using this reduction: If the constraint graph $G$ contains a bisection with $o(|E|)$ edges across, then we can simply run the aforementioned approximation algorithm  on each side of the bisection. This will yield a $(3/4 - o(1))$-approximation. To overcome such an issue, we use a (folklore) result in \NP-hardness of approximation literature that Gap-2-\CSP$_q$\ remains hard to approximate even when $G$ has very good expansion properties. With this, we arrive at the desired result.

\paragraph{\NP-Hardness for GapMinMax-2-\CSP$_q$.}
We use the same reduction as above for proving the \NP-hardness of GapMinMax-2-\CSP$_q$. The completeness of the reduction proceeds in a similar way. The main difference is in the soundness analysis. Roughly speaking, we argue that, in the reconfiguration sequence, we can find a multi-assignment such that, for each vertex, it is assigned either (i) more than one character or (ii) a single character from the original alphabet $\tSigma$. When restricting to case (ii), this gives us a partial assignment that does not violate any constraint of $\tPi$. By starting with known \NP-hardness of approximation results for Clique, we know that there can only little order of variables/vertices involved  in such a case. Thus, the size of the multi-assignment must be at least $(2 - o(1))|V|$ as desired (where $V$ is the set of variables of the GapMinMax-2-\CSP$_q$ instance).

Finally, the Set Cover Reconfiguration hardness follows immediately from applying the ``hypercube gadget'' reduction of Feige~\cite{Feige98}.
\section{Preliminaries}

\paragraph{Notations.} We use the set theoretic notation of $\uplus$ to mean the disjoint union of two sets. For a graph $G = (V, E)$ and any subset $S \subseteq V$, we use $E[S]$ to denote the set of edges whose both endpoints belong to $S$. Meanwhile, for disjoint $S_1, S_2 \subseteq V$, we use $E[S_1, S_2]$ to denote the set of edges whose one endpoint belongs to $S_1$ and the other belongs to $S_2$. Also, we denote by $\deg_G(v)$, the degree of vertex $v\in V$ in the graph $G$.

For any set $S$, let $\cP(S)$ denote the power set of $S$, i.e. the collection of all subsets of $S$. Furthermore, for two sets $S_1, S_2$, we write $S_1 \Delta S_2$ to denote its symmetric difference, i.e. $S_1 \Delta S_2 = (S_1 \setminus S_2) \cup (S_2 \setminus S_1)$.

Let $\Sigma$ be any non-empty set. For every  $d\in\mathbb{N}$ and every pair of strings $x,y\in\Sigma^d$, we denote their Hamming distance by $\|x-y\|_0$ (or equivalently $\|y-x\|_0$) which is defined as:
$$
\|x-y\|_0=\|y-x\|_0=|\{i\in [d]: x_i\neq y_i\}|,
$$
where $x_i$ denotes the character in the $i^{\text{th}}$ position of $x$. 

\subsection{Constraint Satisfaction Problems}\label{sec:prelimCSP}

In this subsection, we define the variants of Constraint Satisfaction Problems (\CSP) relevant to this paper. 

\paragraph{$k$-\CSP.} A $k$-\CSP$_q$ instance $\Pi = (G = (V, E), \Sigma, \{C_e\}_{e \in E})$ consists of:
\begin{itemize}
\item A $k$-uniform hypergraph $G = (V, E)$ called \emph{constraint graph},
\item Alphabet set $\Sigma$ of size at most $q$,
\item For every hyperedge $e = (u_1, \dots, u_k) \in V$, a constraint $C_e: \Sigma^k \to \{0, 1\}$.
\end{itemize}
An \emph{assignment} $\psi$ is a function from $V$ to $\Sigma$. The \emph{value} of $\psi$ is $$\val_{\Pi}(\psi) := \underset{e = (u_1, \dots, u_k) \sim E}{\E}\left[C_e(\psi(u_1), \dots, \psi(u_k))\right].$$ The value of the instance is $\val(\Pi) := \underset{\psi}{\max}\ \val_{\Pi}(\psi)$.

Given two assignments $\psi$ and $\psi'$, we denote their distance by $\|\psi-\psi'\|_0$ (or equivalently $\|\psi'-\psi\|_0$) and is defined as follows:
$$
\|\psi-\psi'\|_0=\|\psi'-\psi\|_0=\lvert \{v\in V: \psi(v)\neq \psi'(v)\}\rvert.
$$

\paragraph{MaxMin $k$-\CSP.}
A \emph{reconfiguration assignment sequence} $\bpsi$ is a sequence $\psi_0, \dots, \psi_p$ of assignments such that $\|\psi_{i - 1} - \psi_i\|_0 = 1$ for all $i \in [p]$. For two assignments $\psi_s$ and $\psi_t$, we write $\bPsi(\psi_s \lrsg \psi_t)$ to denote the set of all reconfiguration assignment sequences starting from $\psi_s$ and ending at $\psi_t$. 

For two assignments $\psi_s$ and $\psi_t$, we say that a sequence is a \emph{direct} reconfiguration assignment sequence from $\psi_s$ to $\psi_t$ if it is a sequence $\bpsi \in \bPsi(\psi_s \lrsg \psi_t)$ such that for every $\psi \in \bpsi$ and every $v \in V$, we have $\psi(v) \in \{\psi_s(v), \psi_t(v)\}$.

For a reconfiguration assignment sequence $\bpsi$, we let $\val_\Pi(\bpsi) = \underset{\psi \in \bpsi}{\min}\ \val_\Pi(\psi)$. Finally, let $\val_\Pi(\psi_s \lrsg \psi_t) = \underset{\bpsi \in \bPsi(\psi_s \lrsg \psi_t)}{\max}\ \val_\Pi(\bpsi)$.

\paragraph{MinLabel 2-\CSP.} A \emph{multi-assignment} is a function $\psi: V \to \cP(\Sigma)$. A multi-assignment $\psi$ is said to satisfy a 2-\CSP\ instance $\Pi$ iff, for every $e = (u, v) \in E$, there exist $\sigma_u \in \psi(u), \sigma_v \in \psi(v)$ such that $C_e(\sigma_u, \sigma_v) = 1$.

Two multi-assignments $\psi_1, \psi_2$ are \emph{neighbors} iff $\underset{v \in V}{\sum} |\psi_1(v) \Delta \psi_2(v)| = 1$. A \emph{reconfiguration multi-assignment sequence} $\bpsi$ is a sequence $\psi_0, \dots, \psi_p$ of multi-assignments such that $\psi_{i - 1}$ and  $\psi_i$ are neighbors for all $i \in [p]$. 
 A reconfiguration multi-assignment sequence $\bpsi = (\psi_0, \dots, \psi_p)$ is said to satisfy $\Pi$ iff $\psi_i$ satisfies $\Pi$ for all $i \in \{0, \dots, p\}$. We write $\Psi^{\SAT(\Pi)}(\psi_s \lrsg \psi_t)$ to denote the set of all satisfying reconfiguration multi-assignment sequence from $\psi_s$ to $\psi_t$.

The \emph{size} of a multi-assignment $\psi$ is defined as  $|\psi| := \underset{v \in V}{\sum} |\psi(v)|$. The size of a reconfiguration multi-assignment sequence $\bpsi$ is defined as $|\bpsi| := \underset{\psi \in \bpsi}{\max}\ |\psi|$.
Finally, the \emph{min-label value} of $\Pi$ from $\psi_s \lrsg \psi_t$ is defined as: $$\minlab_\Pi(\psi_s \lrsg \psi_t) := \underset{\bpsi \in \Psi^{\SAT(\Pi)}(\psi_s \lrsg \psi_t)}{\min} |\bpsi|.$$

\paragraph{Partial Assignments.} We will also use the concept of partial assignments. A partial assignment is defined as $\psi: V \to \Sigma \cup \{\perp\}$ (where $\perp$ can be thought of ``unassigned''). Its size $|\psi|$ is defined as $|\{v\in V \mid \psi(v) \ne \perp\}|$. We say that a partial assignment satisfies $\Pi$ iff, for all $e = (u, v) \in E$ such that $\psi(u)\ne \perp$ and $\psi(v) \ne \perp$, we have $C_e(\psi(u), \psi(v)) = 1$.

We define $\maxpar(\Pi) = \max |\psi|$ where the maximum is over all satisfying partial assignments of $\Pi$.

\paragraph{Gap Problems.} For the purpose of reductions, it will be helpful to work with (promise) gap problems. For any $0 \leq s \leq c \leq 1$, we define the gap problems as follows:
\begin{itemize}
\item In the $\gapcsp{c}{s}{q}$ problem, we are given as input a 2-\CSP$_q$ instance $\Pi$. The goal is to decide if $\val(\Pi) \geq c$ or $\val(\Pi) < s$.
\item In the $\gapmaxmincsp{c}{s}{q}$ problem, we are given as input a $k$-\CSP$_q$ instance together with two assignments $\psi_s$ and  $\psi_t$. The goal is to decide if $\val_\Pi(\psi_s \lrsg \psi_t) \geq c$ or  $\val_\Pi(\psi_s \lrsg \psi_t) < s$.
\item In the $\gapminmaxcsp{1}{1/s}{q}$ problem, we are given as input a 2-\CSP$_q$ instance $\Pi = (G=(V, E), \Sigma, \{C_e\}_{e \in E})$ together with two assignments $\psi_s$ and $\psi_t$. The goal is to decide if\footnote{Note that even when $\val_{\Pi}(\psi_s \lrsg \psi_t) = 1$, we still have $\minlab(\psi_s \lrsg \psi_t) = |V| + 1$.} $\minlab_\Pi(\psi_s \lrsg \psi_t) \leq |V| + 1$ or $\minlab_\Pi(\psi_s \lrsg \psi_t) > (|V| + 1)/s$.
\end{itemize}

\subsection{Set Cover}\label{sec:prelimSetCover}

In the set cover reconfiguration problem, we are given as input subsets $S_1, \dots, S_m \subseteq [n]$. A \emph{reconfiguration set cover sequence} $\bT$ is a sequence $T_0, \dots, T_p$ such that every $T_i$ is a set of indices of a set cover, i.e., each $T_i$ is a subset of $[m]$ and  $\bigcup_{j \in T_i} S_j = [n]$, and moreover, for all $i \in [p]$, we have that $|T_i \Delta T_{i - 1}| = 1$. We are also given as part of the input to the  set cover reconfiguration problem, two set covers $T_s, T_t$ and the goal is  to find a reconfiguration set cover sequence $\bT$ that minimizes $\max_{T \in \bT} |T|$.

\subsection{Error Correcting Codes} \label{sec:code}

A binary \emph{error correcting code (ECC)} of message length $k$  and block length $n$ is an encoding algorithm $\enc: \{0,1\}^k \to \{0,1\}^n$. Its (absolute) distance $d_{\enc}$ is defined as $\underset{s_1 \ne s_2 \in \{0,1\}^k}{\min}\ \|\enc(s_1)- \enc(s_2)\|_0$. The relative distance $\delta_{\enc}$ is defined as $d_{\enc} / n$. Finally, the rate is defined as $k / n$.

It is well known that ECC with constant\footnote{In fact, our proof (of \Cref{thm:rih}) requires only ECCs with \emph{polynomial} rate.} rate and constant relative distance exists.

\begin{theorem}[{\cite[Theorem E.2]{G09}}] \label{thm:ecc}
There exists $\delta, r > 0$ such that, for all $k\in\mathbb{N}$, there exists an encoding $\enc: \{0, 1\}^k \to \{0, 1\}^n$ that is an ECC of relative distance at least $\delta$ and rate at least $r$. Furthermore, $\enc$ runs in polynomial time and there is a   circuit $\dec:\{0, 1\}^n\to \{0, 1\}^k\cup \{\perp\}$   of polynomial size with the following guarantee: 
$$
\forall x\in\{0,1\}^n,\ \dec(x):= 
\begin{cases}
   y&\text{ if }\enc(y)=x\\ 
   \perp&\text{ otherwise}
\end{cases}.$$
\end{theorem}

\subsection{Assignment Testers a.k.a.\ Probabilistically Checkable Proofs of Proximity}\label{sec:pcpp}

\emph{Assignment testers} are the main technical tool used in our proof of \RIH. We remark that assignment testers are equivalent to   \emph{Probabilistically Checkable Proofs of Proximity}  (\PCPP) \cite{BGHSV06}, but we use the term assignment tester here to be consistent with \cite{Dinur07}, whose result we use.

\begin{definition}[Assignment Tester {(\PCPP)}~\cite{DinurR06}]\label{def:asgn}
An \emph{assignment tester} {(\PCPP)} with alphabet set $\Sigma$ (where $q:=|\Sigma|$), {proximity parameter $\kappa$,} and rejection probability $\gamma$ is an algorithm $\cP$ whose input is a Boolean circuit $\Phi$ with input variable set $X$, and whose output is a 2-\CSP$_q$ instance $\Pi = (G = (V, E), \Sigma, \{C_e\}_{e \in E})$ where $V = X \uplus A$ (for some non-empty set $A$) such that the following holds for all assignments $\psi: X \to \{0, 1\}$:
\begin{itemize}
\item (Completeness) If $\psi$ is a satisfying assignment to $\Phi$ (i.e., $\Phi$ on input $\psi$ evaluates to 1), then there exists $\psi^*_A: A \to \Sigma$ such that $\val_{\Pi}((\psi , \psi^*_A)) = 1$. 
\item (Soundness) 
Let $\psi^*:= \underset{\substack{\psi':X\to\{0,1\}\\ \psi' \text{ satisfies }\Phi}}{\argmin}\|\psi-\psi'\|_0$. Then, for every $\psi_A: A \to \Sigma$ we have that 
{$\val_{\Pi}((\psi, \psi_A))<1-\gamma$ whenever $\|\psi-\psi^*\|_0\ge\kappa\cdot|X|$.}
\end{itemize}
\end{definition}

We will use the below construction of assignment testers. 

\begin{theorem}[{\cite[Corollary 9.3]{Dinur07}}] \label{thm:din-pcpp}
{For any constant $\kappa_0>0$, there are some constants $q_0\in\mathbb N,\gamma_0>0$ and} a polynomial-time assignment tester with alphabet size $q_0$, {proximity parameter $\kappa_0$,} and rejection probability $\gamma_0$. 
Furthermore, in the completeness case, there is a polynomial time algorithm $\asgnt$ which takes as input $\psi$  and outputs $\psi^*_A$ (we think of $\asgnt$ as a function that maps $\psi$  to $\psi^*_A$).
\end{theorem}

{
\begin{remark}\label{rmk:pcpp_csp_game_grw}
In the statement of \Cref{thm:main_grw} and \Cref{cor:main_grw}, we use the following standard equivalent perspective of \Cref{def:asgn}.
We view the \CSP\ in \Cref{def:asgn} as a verifier checking a uniformly random constraint against an assignment $\psi$ and an auxiliary proof $\psi_A$.
Then $\val_{\Pi}((\psi, \psi_A))$ equals the probability that the verifier accepts $(\psi,\psi_A)$.
In this formulation, the arity $k$ (e.g., $k=2$ in \Cref{def:asgn}) of the \CSP\ corresponds to the query complexity of the verifier.
\end{remark}
}

\section{\PSPACE-Hardness of Approximation of Reconfiguration}

In this section we prove \RIH\ {(\Cref{thm:rih}) and its quantitative trade-off version (\Cref{thm:main_grw})}.

\begin{proof}[Proof of \Cref{thm:rih}]
It is known that $\gapmaxmintwocsp{1}{1}{q}$ is \PSPACE-hard even for $q=3$ (by putting together \cite{GopalanKMP09} and \cite[Lemma 3.4]{Ohsaka23prev}). This is the starting point of our reduction.

Given a $\gapmaxmintwocsp{1}{1}{q}$ instance $\Pi = (G = (V, E), \Sigma, \{C_e\}_{e \in E})$ with two satisfying assignments $\psi_s,\psi_t:V\to \Sigma$, let $\enc: \{0, 1\}^k \to \{0, 1\}^n$ denote the ECC as guaranteed by \Cref{thm:ecc} of relative distance $\delta$ (along with $\dec:\{0, 1\}^n\to \{0, 1\}^k\cup \{\perp\}$) and where $k = |V| \cdot  \lceil \log q \rceil $. 
Let $q_0, {\kappa_0,}\gamma_0$ be as in \Cref{thm:din-pcpp}; we assume w.l.o.g. that $q_0 \geq 4$ {and $\kappa_0<\delta/4$}. Moreover, let $\pi:\Sigma\to \{0,1\}^{\lceil \log q\rceil}$ be some canonical injective map.  Then, for every $x:=(x_1,\ldots ,x_{|V|})\in\{0,1\}^k$, where for all $i\in|V|$, we have $x_i\in\{0,1\}^{\lceil \log q\rceil}$,  if we have that $\pi^{-1}(x_i)$ exists for all $i\in|V|$,  then we denote by $\psi_{x}:V\to \Sigma$ an assignment to $\Pi$, where $\left(\psi_x(v)\right)_{v\in V}:=(\pi^{-1}(x_1),\ldots, \pi^{-1}(x_{|V|}))$.

For the rest of this proof, for every $i\in[4]$, let $\bar{i}_1, \bar{i}_2,$ and $\bar{i}_3$ denote the elements of $[4] \setminus \{i\}$. Let 
{$\eps=\gamma_0/4$.}
We reduce $\Pi$ to an instance of $\gapmaxminthreecsp{1}{1-\eps}{q_0}$, namely, $\tPi = (\tG = (\tV, \tE), \tSigma, \{\tC_e\}_{e \in \tE})$ where, $$\tV := \{v^*\} \uplus \left(\biguplus_{i \in [4]} \tV_i\right) \uplus \left(\biguplus_{i \in [4]} \tA_i\right)$$
and 
$$\tE := \biguplus_{i \in [4]} \tE_i$$ where the variables and constraints are defined as follows.
\begin{description}
\item[Vertex Set:] First, for all $i \in [4]$, let $\tV_i$ denote a set of $n$ fresh variables.  Next, for all $i \in [4]$:
\begin{itemize}
\item We define  a Boolean circuit $\Phi_i$   on variable set $\tV_{\bar{i}_1}\uplus \tV_{\bar{i}_2}\uplus \tV_{\bar{i}_3}$ by specifying exactly which assignments evaluate it to 1. An assignment $\phi_i:\tV_{\bar{i}_1}\uplus \tV_{\bar{i}_2}\uplus \tV_{\bar{i}_3}\to \{0,1\}$ makes $\Phi_i$ evaluate to 1 if and only if  all of the following holds: 
\begin{description}
    \item[Encoding of a Satisfying Assignment Check:] For all $\ell\in [3]$, we have $\dec(\phi_{i}|_{\tV_{\bar{i}_{\ell}}})\neq \perp$ where $\phi_{i}|_{\tV_{\bar{i}_{\ell}}}$ is the assignment $\phi_i$ restricted to the variables in $\tV_{\bar{i}_{\ell}}$. 
    Moreover, let $x:=\dec(\phi_{i}|_{\tV_{\bar{i}_{\ell}}})$. Then $\psi_x$ satisfies all constraints in $\Pi$.
    \item[Reconfiguration Assignment Sequence Membership Check:] For all $\ell,\ell'\in [3]$ let $x:=\dec(\phi_{i}|_{\tV_{\bar{i}_{\ell}}})$ and $x':=\dec(\phi_{i}|_{\tV_{\bar{i}_{\ell'}}})$. Then, we have that $\|\psi_x-\psi_{x'}\|_0\le 1$.
\end{description}
Note that $\Phi_i$ is efficiently computable.
\item Let $\Pi'_i = (G'_i = (V'_i, E'_i), \tilde\Sigma, \{C^i_e\}_{e \in E'})$ be the 2-\CSP$_{q_0}$ instance produced by applying \Cref{thm:din-pcpp} on $\Phi_i$ where $V'_i = \tV_{\bar{i}_1} \uplus \tV_{\bar{i}_2} \uplus \tV_{\bar{i}_3} \uplus \tA_i$.
\end{itemize}
\item[Hyperedge Set and Constraints:] For all $i\in[4]$, and for each $e = (u, v) \in E'_i$, create a hyperedge $\te = (v^*, u, v)$ in $\tE_i$ with the following constraint:
\begin{align*}
\forall \tsigma^*, \tsigma_u, \tsigma_v\in\tSigma,\ \tC_{\te}(\tsigma^*, \tsigma_u, \tsigma_v) = 1 \Longleftrightarrow  \left(\left(\left(\tsigma^* = i\right) \wedge \left(C^i_{e}(\tsigma_u, \tsigma_v)=1\right)\right) \text{ or }\left(\tsigma^* \in [4] \setminus \{i\}\right)\right).
\end{align*}
\item[Beginning and End of the Reconfiguration Assignment Sequence:] In order to define $\tpsi_s$ and $\tpsi_t$, it will be convenient to first define an additional notion. For every satisfying assignment $\psi: V \to \Sigma$ of $\Pi$, we define an assignment $\psi^{\enc}: \tV \to \tilde\Sigma$ of $\tPi$ in the following way.
We use the shorthand notation,  $\enc(\psi):=\enc((\pi(\psi(v)))_{v\in V})$ throughout the proof. First, fix $i\in[4]$ and we will build an assignment $\phi_i$ to $\Phi_i$ which evaluates it to 1 in the following way:
$$\forall \ell \in [3],\ \phi_i|_{\tV_{\bar{i}_\ell}}:= \enc(\psi). 
$$

From the construction of $\phi_i$ and the assumption that $\psi$ is a satisfying assignment to $\Pi$, it is easy to verify that $\phi_i$ evaluates to 1 on $\Phi_i$. 
Let $\phi^*_{\tA_i}:\tA_i\to \tSigma$ be the output of $\asgnt$ on input $\phi_i$ (as guaranteed   in the completeness case of \Cref{thm:din-pcpp}).

Then, we define $\psi^{\enc}$ as follows:
\begin{align*}
\psi^{\enc}(v^*) &= 4, \\
\psi^{\enc}|_{\tV_i} &= \enc(\psi) &\forall i \in [4], \\
\psi^{\enc}|_{\tA_i} &= \asgnt(\phi_i)=\phi^*_{\tA_i} &\forall i \in [4].
\end{align*}
 
  Then, let $\tpsi_s = (\psi_s)^{\enc}$ and $\tpsi_t = (\psi_t)^{\enc}$ and it can be verified that both are satisfying assignments to $\tPi$. 
\end{description}

It is easy to note that the total reduction runs in polynomial time. The rest of the proof is dedicated to showing the completeness and soundness of the reduction.

\paragraph{Completeness Analysis.}
Suppose that $\val_{\Pi}(\psi_s \lrsg \psi_t) = 1$. That is, there exists a reconfiguration assignment sequence $\psi_0, \dots, \psi_p$ (w.r.t. $\Pi$) such that $\psi_0 = \psi_s, \psi_p = \psi_t$ and $\val_{\Pi}(\psi_i) = 1$ for all $i \in \{0, \dots, p\}$. We will show that $\val_{\tPi}(\tpsi_s \lrsg \tpsi_t) = 1$. To do so, it suffices to show that, for any two satisfying assignments $\psi, \psi'$ of $\Pi$ that differ on a single coordinate, we have that  $\val_{\tPi}((\psi)^{\enc} \lrsg (\psi')^{\enc}) = 1$. (After which, we can just concatenate the configuration sequences from $(\psi_0)^{\enc}$ to $(\psi_1)^{\enc}$ and then from $(\psi_1)^{\enc}$ to $(\psi_2)^{\enc}$ and so on.)

Suppose that $\psi$ and  $\psi'$ are two satisfying assignments of $\Pi$ such that $\|\psi-\psi'\|_0=1$.  We now make an important remark. Let $A \subseteq \{0, 1\}^n \times \{0, 1\}^n \times \{0, 1\}^n$ be defined as follows:
$$
(a_1,a_2,a_3)\in A\iff \forall\ \iota\in[3], \text{ we have }a_{\iota}\in \{\enc(\psi),\enc(\psi')\}.
$$
Note that $A$ is of size 8. For every $a\in A$ and for every $i\in [4]$, we have that $\Phi_i$ evaluates to 1 on $a$ because, $\|\psi-\psi'\|_0=1$ (thus passing the Reconfiguration Assignment Sequence Membership Check) and both $\psi$ and  $\psi'$ are  satisfying assignments of $\Pi$ (thus passing the Encoding of a Satisfying Assignment Check). Thus, we can run $\asgnt$ on $a$.

We think of an assignment to $\tPi$ now as a string in: $$\tSigma^{1+\left(\sum_{i\in[n]}|\tV|_i\right)+\left(\sum_{i\in[n]}|\tA|_i\right)}=\tSigma\times \tSigma^{|\tV_1|}\times\tSigma^{|\tV_2|}\times\tSigma^{|\tV_3|}\times\tSigma^{|\tV_4|}\times\tSigma^{|\tA_1|}\times\tSigma^{|\tA_2|}\times\tSigma^{|\tA_3|}\times\tSigma^{|\tA_4|}.$$ Then, we consider a sequence of assignments $\hat\bpsi$  (which is not a reconfiguration assignment sequence) of $\tPi$ given as follows: $$\hat\bpsi:=\langle\hat\psi_0,\ldots ,\hat\psi_{8}\rangle,$$
where we have:
\begin{align*}
    \hat{\psi}_0&=\psi^{\enc}=(4,\psi^{\enc}|_{\tV_1},\psi^{\enc}|_{\tV_2},\psi^{\enc}|_{\tV_3},\psi^{\enc}|_{\tV_4},\psi^{\enc}|_{\tA_1},\psi^{\enc}|_{\tA_2},\psi^{\enc}|_{\tA_3},\psi^{\enc}|_{\tA_4}),\\
     \hat{\psi}_1&= (\mathcolorbox{pink!50!white}{1},\psi^{\enc}|_{\tV_1},\psi^{\enc}|_{\tV_2},\psi^{\enc}|_{\tV_3},\psi^{\enc}|_{\tV_4},\psi^{\enc}|_{\tA_1},\psi^{\enc}|_{\tA_2},\psi^{\enc}|_{\tA_3},\psi^{\enc}|_{\tA_4}),\\
     \hat{\psi}_2&=\Big(1,\mathcolorbox{pink!50!white}{(\psi')^{\enc}|_{\tV_1}},\psi^{\enc}|_{\tV_2},\psi^{\enc}|_{\tV_3},\psi^{\enc}|_{\tV_4},\psi^{\enc}|_{\tA_1},\mathcolorbox{pink!50!white}{\asgnt((\psi')^{\enc}|_{\tV_1},\psi^{\enc}|_{\tV_3},\psi^{\enc}|_{\tV_4})},\\
     &\phantom{10305630}\mathcolorbox{pink!50!white}{\asgnt((\psi')^{\enc}|_{\tV_1},\psi^{\enc}|_{\tV_2},\psi^{\enc}|_{\tV_4})},\mathcolorbox{pink!50!white}{\asgnt((\psi')^{\enc}|_{\tV_1},\psi^{\enc}|_{\tV_2},\psi^{\enc}|_{\tV_3})}\Big),\\
      \hat{\psi}_3&= \Big(\mathcolorbox{pink!50!white}{2} ,{(\psi')^{\enc}|_{\tV_1}},\psi^{\enc}|_{\tV_2},\psi^{\enc}|_{\tV_3},\psi^{\enc}|_{\tV_4},\psi^{\enc}|_{\tA_1},{ \asgnt((\psi')^{\enc}|_{\tV_1},\psi^{\enc}|_{\tV_3},\psi^{\enc}|_{\tV_4})},\\
     &\phantom{10305630}{ \asgnt((\psi')^{\enc}|_{\tV_1},\psi^{\enc}|_{\tV_2},\psi^{\enc}|_{\tV_4})},{ \asgnt((\psi')^{\enc}|_{\tV_1},\psi^{\enc}|_{\tV_2},\psi^{\enc}|_{\tV_3})}\Big),\\
      \hat{\psi}_4&= \Big({2} ,{(\psi')^{\enc}|_{\tV_1}},\mathcolorbox{pink!50!white}{(\psi')^{\enc}|_{\tV_2}},\psi^{\enc}|_{\tV_3},\psi^{\enc}|_{\tV_4},\\
      &\phantom{10305630}\mathcolorbox{pink!50!white}{\asgnt((\psi')^{\enc}|_{\tV_2},\psi^{\enc}|_{\tV_3},\psi^{\enc}|_{\tV_4})},{ \asgnt((\psi')^{\enc}|_{\tV_1},\psi^{\enc}|_{\tV_3},\psi^{\enc}|_{\tV_4})},
\\
&\phantom{10305630}\mathcolorbox{pink!50!white}{\asgnt((\psi')^{\enc}|_{\tV_1},(\psi')^{\enc}|_{\tV_2},\psi^{\enc}|_{\tV_4})},
     \mathcolorbox{pink!50!white}{\asgnt((\psi')^{\enc}|_{\tV_1},(\psi')^{\enc}|_{\tV_2},\psi^{\enc}|_{\tV_3})}\Big),\\
      \hat{\psi}_5&= \Big(\mathcolorbox{pink!50!white}{3} ,{(\psi')^{\enc}|_{\tV_1}},{(\psi')^{\enc}|_{\tV_2}},\psi^{\enc}|_{\tV_3},\psi^{\enc}|_{\tV_4},\\
      &\phantom{10305630}{ \asgnt((\psi')^{\enc}|_{\tV_2},\psi^{\enc}|_{\tV_3},\psi^{\enc}|_{\tV_4})},{ \asgnt((\psi')^{\enc}|_{\tV_1},\psi^{\enc}|_{\tV_3},\psi^{\enc}|_{\tV_4})},\\
     &\phantom{10305630}
     {  \asgnt((\psi')^{\enc}|_{\tV_1},(\psi')^{\enc}|_{\tV_2},\psi^{\enc}|_{\tV_4})},{  \asgnt((\psi')^{\enc}|_{\tV_1},(\psi')^{\enc}|_{\tV_2},\psi^{\enc}|_{\tV_3})}\Big),\\
      \hat{\psi}_6&\Big({3} ,{(\psi')^{\enc}|_{\tV_1}},{(\psi')^{\enc}|_{\tV_2}},\mathcolorbox{pink!50!white}{(\psi')^{\enc}|_{\tV_3}},\psi^{\enc}|_{\tV_4},\\
      &\phantom{10305630}\mathcolorbox{pink!50!white}{\asgnt((\psi')^{\enc}|_{\tV_2},(\psi')^{\enc}|_{\tV_3},\psi^{\enc}|_{\tV_4})},\mathcolorbox{pink!50!white}{\asgnt((\psi')^{\enc}|_{\tV_1},(\psi')^{\enc}|_{\tV_3},\psi^{\enc}|_{\tV_4})},\\
     &\phantom{10305630}{  \asgnt((\psi')^{\enc}|_{\tV_1},(\psi')^{\enc}|_{\tV_2},\psi^{\enc}|_{\tV_4})},\mathcolorbox{pink!50!white}{\asgnt((\psi')^{\enc}|_{\tV_1},(\psi')^{\enc}|_{\tV_2},(\psi')^{\enc}|_{\tV_3})}\Big),\\
      \hat{\psi}_7&= \Big(\mathcolorbox{pink!50!white}{4} ,{(\psi')^{\enc}|_{\tV_1}},{(\psi')^{\enc}|_{\tV_2}},{ (\psi')^{\enc}|_{\tV_3}},\psi^{\enc}|_{\tV_4},\\
      &\phantom{10305630}{  \asgnt((\psi')^{\enc}|_{\tV_2},(\psi')^{\enc}|_{\tV_3},\psi^{\enc}|_{\tV_4})},{  \asgnt((\psi')^{\enc}|_{\tV_1},(\psi')^{\enc}|_{\tV_3},\psi^{\enc}|_{\tV_4})},
     \\
     &\phantom{10305630}{  \asgnt((\psi')^{\enc}|_{\tV_1},(\psi')^{\enc}|_{\tV_2},\psi^{\enc}|_{\tV_4})},
     {  \asgnt((\psi')^{\enc}|_{\tV_1},(\psi')^{\enc}|_{\tV_2},(\psi')^{\enc}|_{\tV_3})}\Big),\\
      \hat{\psi}_8&= \Big({4} ,{(\psi')^{\enc}|_{\tV_1}},{(\psi')^{\enc}|_{\tV_2}},{(\psi')^{\enc}|_{\tV_3}},\mathcolorbox{pink!50!white}{(\psi')^{\enc}|_{\tV_4}},\\
     &\phantom{10305630}\mathcolorbox{pink!50!white}{\asgnt((\psi')^{\enc}|_{\tV_2},(\psi')^{\enc}|_{\tV_3},(\psi')^{\enc}|_{\tV_4})},\mathcolorbox{pink!50!white}{\asgnt((\psi')^{\enc}|_{\tV_1},(\psi')^{\enc}|_{\tV_3},(\psi')^{\enc}|_{\tV_4})},\\
     &\phantom{10305630}\mathcolorbox{pink!50!white}{\asgnt((\psi')^{\enc}|_{\tV_1},(\psi')^{\enc}|_{\tV_2},(\psi')^{\enc}|_{\tV_4})},{  \asgnt((\psi')^{\enc}|_{\tV_1},(\psi')^{\enc}|_{\tV_2},(\psi')^{\enc}|_{\tV_3})}\Big)\\&=(\psi')^{\enc},
\end{align*}
where we have highlighted in pink color the entries which have changed from the immediate predecessor in the sequence. It is easy to verify that each assignment in $\hat\bpsi$ is a satisfying assignment to $\tPi$.

Next, an important observation is that for all $j\in[8]$,  let $\hat{\bpsi}_j$, be the \emph{direct} reconfiguration assignment sequence from $\hat{\psi}_{j-1}$ to $\hat{\psi}_j$ obtained by changing  $\hat\psi_{j-1}$ to $\hat\psi_j$ in some canonical shortest way possible. Then,  for every intermediate assignment $\hat{\psi}_{j-1,j}\in \hat{\bpsi}_j$  (i.e., $\|\hat{\psi}_{j-1,j}-\hat{\psi}_{j-1}\|_0+\|\hat{\psi}_{j-1,j}-\hat{\psi}_{j}\|_0=\|\hat{\psi}_{j-1}-\hat{\psi}_{j}\|_0$), we have that $\hat{\psi}_{j-1,j}$ is also a \emph{satisfying} assignment. This is because, if $j$ is even then the variables whose assignment is modified in $\hat{\psi}_{j-1,j}$ (compared to $\hat{\psi}_{j-1}$) is not in $\Pi_i'$,   where $i$ is the first coordinate of $\hat{\psi}_{j-1,j}$, and if $j$ is odd then note that for all $j'\in[9]$, we have by construction of the assignment that in $\hat{\psi}_{j'-1}$ if we change the assignment of $v^*$ to any arbitrary value in $[4]$, the modified assignment continues to satisfy all constraints of $\tPi$. 

Thus, our reconfiguration assignment sequence $\bpsi$ from $(\psi)^{\enc}$ to $(\psi')^\enc$ is as follows. Set $j=0$, and while $j<8$, first include $\hat{\psi}_j$ to $\bpsi$, and then sequentially introduce the assignments in the   the \emph{direct} reconfiguration assignment sequence $\hat{\bpsi}_{j+1}$ which ends with the assignment $\hat{\psi}_{j+1}$.  

It is simple to verify that the above construction gives a valid reconfiguration assignment sequence from $(\psi)^{\enc}$ to $(\psi')^{\enc}$ such that every assignment is satisfying all constraints in $\tPi$.

\paragraph{Soundness Analysis.}
Suppose contrapositively that $\val_{\tilde \Pi}(\tpsi_s \lrsg \tpsi_t) \geq 1 - \eps$. We will show that $\val_{\Pi}(\psi_s \lrsg \psi_t) = 1$. 

Since $\val_{\tilde \Pi}(\tpsi_s \lrsg \tpsi_t) \geq 1 - \eps$, there exists a reconfiguration assignment sequence $\tpsi_0, \dots, \tpsi_p$ (w.r.t. $\tPi$) such that $\tpsi_0 = \tpsi_s, \tpsi_p = \tpsi_t$ and $\val_{\tPi}(\tpsi_j) \geq 1 - \eps$ for all $j \in \{0, \dots, p\}$. We further assume that for all $j \in \{0, \dots, p\}$, we have $\tpsi_j(v^*)\in[4]$, as otherwise we have all the constraints are violated in $\tPi$.

For every $j \in \{0, \dots, p\}$, we construct an assignment $\psi_j$ to $\Pi$ as follows:
\begin{itemize}
\item Let $i_j := \tpsi_j(v^*)$.
\item For each $\ell \in [4] \setminus \{i_j\}$, let $\psi^{\ell}_j := \underset{\psi:V\to\Sigma}{\argmin} \ \|\enc(\psi)- \tpsi_{j}|_{\tV_{\ell}}\|_0$.
\item For every $v \in V$, let $\psi_j(v)$ be the most frequent element in $\{\psi^{\ell}_j(v)\}_{\ell \in [4] \setminus \{i_j\}}$ (ties broken arbitrarily).
\end{itemize}

Our main   observation is the following claim:
\begin{claim} \label{obs:dec}
For every $j \in \{0,\ldots ,p\}$ the following holds.
\begin{enumerate}
\item For every $\ell \in [4] \setminus \{i_j\}$, we have $\|\enc(\psi_j^{\ell})- \tpsi_{j}|_{\tV_{\ell}}\|_0 < \frac{\delta n}{4}  $.
\item For every $\ell\in [4] \setminus \{i_j\}$, we have $\psi^{\ell}_j$ is a satisfying assignment of $\Pi$.
\item For every $\ell, \ell' \in [4] \setminus \{i_j\}$, we have $\|\psi^{\ell}_j- \psi^{\ell'}_j\|_0 \leq 1$.
\end{enumerate}
\end{claim}

Before we prove \Cref{obs:dec},
let us first see how the above claim helps us finish the soundness analysis. 

Note that, for all  $j\in\{0,\ldots ,p\}$, due to the third item in \Cref{obs:dec}, at least two of the three assignments among $\psi^{\ell}_j$s (for $\ell \in [4] \setminus \{i_j\}$) are identical. Since $\psi_j$ will be equal to this, the second item of \Cref{obs:dec} implies that $\val_{\Pi}(\psi_j) = 1$.

Fix some $j\in[p]$. We claim  that  $\|\psi_j-\psi_{j-1}\|_0\le 1$. To see this consider two cases: 
\begin{description}
    \item[Case I:] $\tpsi_{j - 1}(v^*) = \tpsi_j(v^*)$. In this case, from the first item of \Cref{obs:dec} we can deduce $\psi_{j - 1}^\ell = \psi_j^\ell$, for all $\ell \in [4] \setminus \{i_j\}$. This is because, for any fixing of $\ell \in [4] \setminus \{i_j\}$ we have from triangle inequality:
    \begin{align*}
&        \left\|\enc(\psi_j^{\ell})- \enc(\psi_{j-1}^{\ell})\right\|_0 \\
\le\ & \left\|\enc( \psi_{j-1}^{\ell})- \tpsi_{j-1}|_{\tV_{\ell}}\right\|_0 +\left\|\enc( \psi_j^{\ell})- \tpsi_{j}|_{\tV_{\ell}}\right\|_0 +\left\|\tpsi_{j}|_{\tV_{\ell}}-\tpsi_{j-1}|_{\tV_{\ell}}\right\|_0\\
\le\ & \frac{\delta n}{2}+1.
    \end{align*}

    Since $d_{\enc}\ge \delta n$, this implies that for all $v\in V$, we have  $\pi(\psi_j^{\ell}(v))= \pi(\psi_{j-1}^{\ell}(v))$, which further implies that $\psi_j^{\ell}=\psi_{j-1}^{\ell}$ (as $\pi$ is injective). Since this last equality holds for all $\ell \in [4] \setminus \{i_j\}$, we thus have $\psi_{j - 1} = \psi_j$. 
\item[Case II:] $\tpsi_{j - 1}(v^*) \ne \tpsi_j(v^*)$. Let $\bar{i}_1$ and $\bar{i}_2$ denote the elements in $[4] \setminus \{\tpsi_{j - 1}(v^*), \tpsi_j(v^*)\}$. 
From the construction of $\psi_{j -1}$ and $ \psi_j$, we have that $\psi_{j -1}\in \{\psi_{j -1}^{\bar{i}_1},\psi_{j -1}^{\bar{i}_2}\}$ and $\psi_{j }\in \{\psi_{j }^{\bar{i}_1},\psi_{j }^{\bar{i}_2}\}$. However, since $\|\tpsi_{j - 1}- \tpsi_j\|_0=1$ we have that $\tpsi_{j - 1}|_{\tV\setminus \{v^*\}} = \tpsi_j|_{\tV\setminus \{v^*\}}$. This implies that $\psi_{j-1 }^{\bar{i}_1}=\psi_{j }^{\bar{i}_1}$ and $\psi_{j-1 }^{\bar{i}_2}=\psi_{j }^{\bar{i}_2}$. Thus, $\psi_{j -1}\in \{\psi_{j }^{\bar{i}_1},\psi_{j }^{\bar{i}_2}\}$. But we know from Item 3 in \Cref{obs:dec} that $\|\psi_{j }^{\bar{i}_1}-\psi_{j }^{\bar{i}_2}\|_0\le 1$. Thus, we conclude that $\|\psi_{j }-\psi_{j -1}\|_0\le 1$.
\end{description}

Now, consider the sequence of assignments, $\langle \psi_0, \dots, \psi_p\rangle$. We remove all contiguous duplicates, i.e., for all $j\in[p]$, if $\psi_j=\psi_{j-1}$, then we remove $\psi_j$ from the sequence. Let the resulting sequence be $\bpsi$. It is easy to see that $\bpsi$ is a valid reconfiguration assignment sequence. This completes the proof. 

\begin{proof}[Proof of \Cref{obs:dec}]
Fix $j\in \{0,\ldots ,p\}$. For ease of notation, we will use the shorthand $i$ for $i_j$ in this proof.  Since $\val_{\tPi}(\tpsi_j) \geq 1 - \eps$ we have that at most $4\eps \cdot |\tE_{i}|$ many constraints are violated by $\tpsi_j$ in $\{\tC_{\tilde{e}}\}_{\te\in\tE_{i}}$,  which in turn implies that $\tpsi_j|_{V'_{i}}$ violates at most $4\eps \cdot |E'_{i}|$ many constraints in $\Pi'_{i}$. 
Let $\hat V_i:=\tV_{\bar{i}_1}\uplus \tV_{\bar{i}_2}\uplus \tV_{\bar{i}_3}$. 
Let $\psi^*:\hat{V}_i\to\{0,1\}$ be defined as follows:
$$\psi^*:=\underset{\substack{\psi':\hat{V}_i\to\{0,1\}\\ \psi' \text{ satisfies }\Phi_i}}{\argmin}\|\tpsi_j|_{\hat{V}_{i}}-\psi'\|_0.$$

From the soundness guarantee of \Cref{def:asgn} we have
{
$$
1-\gamma_0\le1-4\eps\le\val_{\Pi'_i}(\tpsi_j|_{V'_{i}})
\implies \left\|\tpsi_j|_{\hat{V}_{i}}-\psi^*\right\|_0\le \kappa_0\cdot n<\frac{\delta}{4}\cdot n.$$
}

Fix  $\ell\in [4]\setminus \{i\}$. We then have: 
$$\left\|\tpsi_j|_{{\tV}_{\ell}}-\psi^*|_{{\tV}_{\ell}}\right\|_0\le\left\|\tpsi_j|_{\hat{V}_{i}}-\psi^*\right\|_0 <\frac{\delta}{4}\cdot n.$$

 Since $\psi^*$ satisfies $\Phi_i$ we have that $\dec(\psi^*|_{\tV_{\ell}})\neq \perp$. Let $x^{\ell}:=\dec(\psi^*|_{\tV_{\ell}})$.  We also have that $\enc(\psi_{x^{\ell}})=\psi^*|_{{\tV}_{\ell}}$. From the definition of $\psi^{\ell}_j$ we have that:
 $$\left\|\enc(\psi_j^{\ell})- \tpsi_{j}|_{\tV_{\ell}}\right\|_0\le \left\|\enc( \psi_{x^{\ell}})- \tpsi_{j}|_{\tV_{\ell}}\right\|_0=\left\|\tpsi_j|_{{\tV}_{\ell}}-\psi^*|_{{\tV}_{\ell}}\right\|_0<\frac{\delta}{4}\cdot n$$

Thus we proved Item 1 of the claim statement. Next to prove Item 2, observe that from triangle inequality we have:
\begin{align*}
&\left\|\enc( \psi_j^{\ell})- \enc( \psi_{x^{\ell}})\right\|_0\\
\le & \left\|\enc( \psi_j^{\ell})- \tpsi_j|_{{\tV}_{\ell}}\right\|_0+\left\|\tpsi_j|_{{\tV}_{\ell}}- \enc( \psi_{x^{\ell}})\right\|_0\\
<&\frac{\delta}{2}\cdot n.
\end{align*}
    Since $d_{\enc}\ge \delta n$, this implies that for all $v\in V$, we have  $\pi(\psi_j^{\ell}(v))= \pi(\psi_{x^{\ell}}(v))$, which further implies that $\psi_j^{\ell}=\psi_{x^{\ell}}$ (as $\pi$ is injective).  Since, $\psi_{x^{\ell}}$ satisfies all  the constraints in $\Pi$ (because $\psi^*$ satisfies $\Phi_i$ and thus passed the Encoding of a Satisfying Assingment Check), we have that $\psi_j^{\ell}$ satisfies all the constraints in $\Pi$. 

    Finally, to prove Item 3 of the claim statement, fix  some $\ell'\in [4]\setminus \{i\}$. Let $x^{\ell'}:=\dec(\psi^*|_{\tV_{\ell'}})$ and $\enc(\psi_{x^{\ell'}})=\psi^*|_{{\tV}_{\ell'}}$. Then, since $\psi^*$ passed the Reconfiguration Assignment Sequence Membership Check in $\Phi_i$, we have $1 \geq \|\psi_{x^{\ell}}-\psi_{x^{\ell'}}\|_0 = \|\psi_j^{\ell} - \psi_j^{\ell'}\|_0$ as desired.
\end{proof}
\end{proof}

\subsection{The Soundness and Arity Trade-Off}

{
In this section, we we apply the parallelization trick to the above construction to get a better tradeoff between query complexity and soundness (\Cref{thm:main_grw}). 

Recall \Cref{rmk:pcpp_csp_game_grw}. We will work with the more convenient proof-verifier perspective of the assignment tester and refer to the \CSP\ instance $\Phi$ in \Cref{def:asgn} as verifier $\mathcal V$.

\begin{definition}[Parallelizable Assignment Tester]\label{def:paraPCPP_grw}
    For each $i\in[t]$, assume $\mathcal V_i$ is a $k$-query assignment tester for Boolean circuit $\Phi_i$ on length-$n$ assignment $\psi^{(i)}$ and length-$m$ auxiliary proof $\psi_A^{(i)}$ with alphabet $\Sigma$.
    For each randomness $r$, we use $I_r^{(i)}$ to denote the $k$ locations that the verifier $\mathcal V_i$ queries on $\psi^{(i)}\circ\psi_A^{(i)}$ given the choice of $r$.
    
    We say these $\mathcal V_i$'s are parallelizable if $I_r^{(i)}=I_r^{(j)}$ holds for every $i,j\in[t]$ and randomness $r$, i.e., they query the same set of $k$ locations in their respective $\psi\circ\psi_A$.
\end{definition}

We think of the input of the $t$ verifiers as written in a $t \times (n+m)$ table:
\begin{itemize}
    \item We use $x_1,\ldots,x_n,\pi_1,\ldots,\pi_m$ to denote the variables on each column, which are supposed to take values in $\Sigma^t$.
    \item We use $x^{(1)},\ldots,x^{(t)}, \pi^{(1)},\ldots,\pi^{(t)}$ to denote the variables on each row, where $x^{(i)}$ is supposed to take values in $\Sigma^n$ and $\pi^{(i)}$ is supposed to take values in $\Sigma^m$, corresponding to $\psi^{(i)}$ and $\psi_A^{(i)}$ respectively.
    \item We abbreviate variables $(x_1, \ldots, x_n)$ and $(\pi_1, \ldots, \pi_m)$ as $x$ and $\pi$ for assignment and auxiliary proof respectively.
\end{itemize}

\begin{theorem}\label{thm:pcpp2csp_grw}
    For any constant $t\ge 1$, let $\mathcal V_1,\ldots,\mathcal V_t$ be $t$ parallelizable $k$-query assignment testers for circuits $\Phi_1,\ldots,\Phi_t$ with length-$n$ assignment, length-$m$ auxiliary proof, and alphabet $\Sigma$ of size $q\ge2$.
    There is a $(k+1)\text{-}\CSP_{q^t}$ instance $\Pi$ which can be built in polynomial time, such that:
    \begin{itemize}
        \item (Variables and Alphabet) 
        $\Pi$ has $n+m+1$ variables $\{v^*\} \cup x \cup \pi$ over alphabet $\Sigma^t$.
        \item (Completeness) Fix an assignment $\psi$ of $\Pi$. 
        If for some $i \in [t]$, $\mathcal V_i$ accepts $\psi(x^{(i)}) \circ \psi(\pi^{(i)})$ with probability 1, then $\val_\Pi(\psi)=1$, where $\psi(x^{(i)}),\psi(\pi^{(i)})$ are the assignments of $x^{(i)},\pi^{(i)}$ in $\psi$ respectively.
        \item (Soundness) Fix an assignment $\psi$ of $\Pi$. If for every $i \in [t]$, $\mathcal V_i$ rejects $\psi(x^{(i)}) \circ \psi(\pi^{(i)})$ with probability at least $\gamma$, then $\val_{\Pi}(\psi) \le1-\gamma$.
    \end{itemize}
\end{theorem}
\begin{proof}
    $\Pi$ is constructed by simply ``parallelizing'' the $t$ verifiers.
    See \Cref{fig:1_grw} for an example.
    
    Specifically, for each randomness in $r$, let $I$ be the size-$k$ local window that every $\mathcal V_i$ queries on their respective $x \circ \pi$. We add a $(k+1)$-ary constraint $c$ to $\Pi$ as follows:
    \begin{itemize}
        \item $c$ is on variables $\{v^*\} \cup (x \circ \pi)_I$.
        \item $c$ is satisfied if and only if:
        \begin{itemize}
            \item $v^*$ takes value in $[t]$ (We treat $\Sigma^t$, the alphabet of $v^*$, as a super-set of $[t]$).
            \item Suppose $v^*$ takes value $i \in [t]$, then $\mathcal V_i$ accepts the assignment on $(x^{(i)} \circ \pi^{(i)})_I$.
        \end{itemize}
    \end{itemize}
    The completeness is obvious.
    The soundness is due to the parallelization trick that we can check the same constraint for all verifier at once.
\end{proof}

\begin{figure}[htbp]
  \centering
  \includegraphics[width=0.8\textwidth]{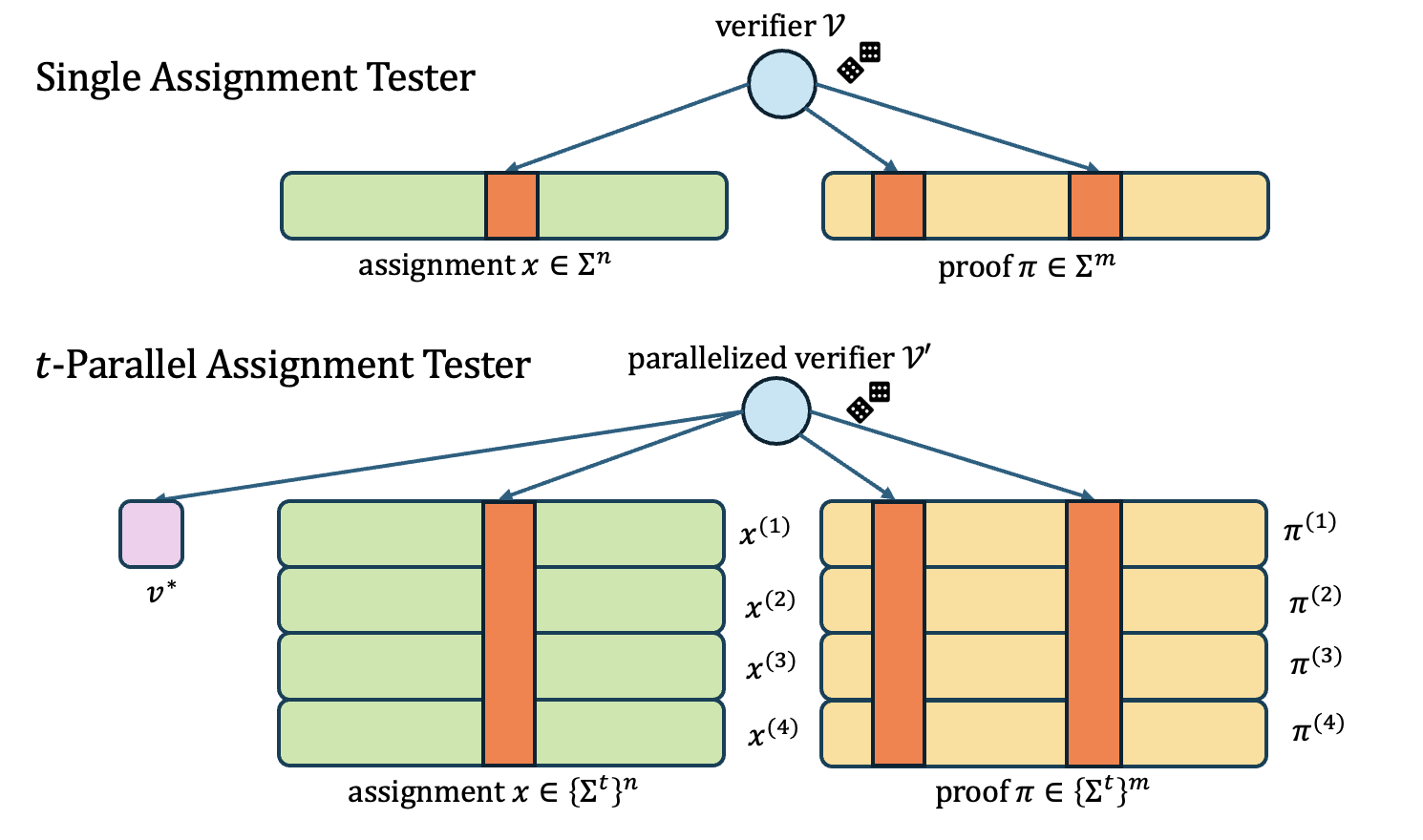}
  \caption{An illustration of the parallel construction with $t=4$. 
  }
  \label{fig:1_grw}
\end{figure}

With \Cref{thm:pcpp2csp_grw}, we now prove \Cref{thm:main_grw}.

\begin{proof}[Proof of \Cref{thm:main_grw}]
As mentioned in \Cref{sec:proof-overview}, the proof of \Cref{thm:rih} actually implies the \PSPACE-hardness of the following reconfiguration task: given four parallelizable $k$-query assignment testers $\mathcal V_1,\mathcal V_2,\mathcal V_3,\mathcal V_4$ and two assignments $\psi^{\textsf{ini}},\psi^{\textsf{tar}}:(x \cup \pi) \to \Sigma^4$, distinguish between the following two cases:
    \begin{itemize}
        \item (Yes Case) There is a reconfiguration assignment sequence $\psi^{\textsf{ini}}\lrsg\psi^\textsf{tar}$, such that for every intermediate assignment $\psi$, there exists $i \in [t]$ such that $\mathcal V_i$ accepts $\psi(x^{(i)})\circ \psi(\pi^{(i)})$ with probability $1$.
        \item (No Case) For any reconfiguration assignment sequence $\psi^\textsf{ini}\lrsg\psi^\textsf{tar}$, there exists an intermediate assignment $\psi$, such that for all $i \in [t]$, $\mathcal V_i$ accepts $\psi(x^{(i)}) \circ \psi(\pi^{(i)})$ with probability less than $1-\gamma$.
    \end{itemize}
Here $\gamma$ is the rejection probability of each assignment tester.

Note that in the above formulation, we only need $\delta>\kappa$ (instead of $\delta>4\kappa$ as in the original proof), where $\kappa$ is the proximity parameter of the assignment testers and $\delta$ is the relative distance of the ECC used in the proof of \Cref{thm:rih}.
Since $\kappa<1$, this is not an issue.

Let $\Sigma$ be the alphabet of the above assignment testers and define $q=|\Sigma|$.
By \Cref{thm:pcpp2csp_grw}, the above task is reduced to a $(k+1)\text{-}\CSP_{q^4}$ instance $\Pi$, for which we need to distinguish between the following two cases:
    \begin{itemize}
        \item (Yes Case) There is a reconfiguration assignment sequence $\psi^{\textsf{ini}}\lrsg\psi^\textsf{tar}$, such that for every intermediate assignment $\psi$ is a solution to $\Pi$.
        \item (No Case) For any reconfiguration assignment sequence $\psi^\textsf{ini}\lrsg\psi^\textsf{tar}$, there exists an intermediate assignment $\psi$, such that $\val_\Pi(\psi)<1-\gamma$.
    \end{itemize}
This problem is precisely $\gapmaxminvarcsp{1}{1-\gamma}{(k+1)}{q^4}$ and it completes the proof by noticing $q=O(1)$.
\end{proof}
}

\section[NP-Hardness of Approximation with Tight Ratios]{\NP-Hardness of Approximation with Tight Ratios}

In this section we prove \Cref{thm:maxmin-csp-np-hardness}, \Cref{thm:minmax-setcover-np-hardness},~and~ \Cref{thm:minmax-csp-np-hardness}, i.e., our tight \NP-hardness results. 

\subsection{\NP-Hardness of GapMaxMin-2-\CSP$_q$}

\begin{sloppypar}
    In this subsection, we will prove our (nearly) tight \NP-hardness of approximation of GapMaxMin-2-\CSP$_q$~(\Cref{thm:maxmin-csp-np-hardness}). 
\end{sloppypar}

We will reduce from the \NP-hardness of Gap-2-\CSP$_q$ problem with ``balanced'' edges. 
\begin{definition} \label{def:balance}
We say that a graph $G = (V, E)$ is \emph{$\delta$-balanced} if and only if for any partition $V = V_1 \cup V_2$ such that $|V_1|, |V_2| \leq \lceil |V| / 2\rceil$, we have $|E[V_1]| + |E[V_2]| \leq (1 + \delta) |E| / 2$. We say that a 2-\CSP\ instance is $\delta$-balanced if and only if it's constraint graph is $\delta$-balanced.
\end{definition}

For $\delta$-balanced 2-\CSP{}s, i.e., instances of Gap-2-\CSP$_q$ whose constraint graph is $\delta$-balanced (for some $\delta>0$), it is not hard to show the following result:

\begin{theorem} \label{thm:balanced-2csp}
For every constant $\delta > 0$, there exists $q \in \N$ such that $\gapcsp{1}{\delta}{q}$ is \NP-hard even on $\delta$-balanced instances.
\end{theorem}

The proof of the above theorem is deferred to \Cref{subsec:balance}.

\subsubsection[From Balanced Instances of Gap-2-CSP to GapMaxMin-2-CSP]{From Balanced Instances of Gap-2-\CSP$_q$ to GapMaxMin-2-\CSP$_q$: Proof of \Cref{thm:maxmin-csp-np-hardness}}

{Before we prove \Cref{thm:balanced-2csp}, let us  show how to use it to prove \Cref{thm:maxmin-csp-np-hardness}. We note here that our reduction below is slightly different from the one presented in \Cref{sec:proof-overview}. Specifically, we do not add two new characters to the alphabet set, but instead we introduce a new character $\sigma^*$ and then  for each character $\tsigma \in \tSigma\cup \{\sigma^*\}$, we  make two copies of it $(\tsigma, 0)$ and $(\tsigma, 1)$, and we set the constraint so that $(\sigma^*, i)$ is \emph{not} compatible with $(\tsigma, 1 - i)$. This change helps avoid some strategy, such as changing a fraction of assignments to some character from $\tSigma$ before using the approximation algorithm on the remaining assignments, that can prevent the reduction in \Cref{sec:proof-overview} from showing $1/2 + o(1)$ factor hardness of approximation.}

\begin{proof}[Proof of \Cref{thm:maxmin-csp-np-hardness}]
\begin{sloppypar}Let $\delta = \eps/2$. For any two bits $a,b\in\{0,1\}$, we define the indicator function   $\ind[a,b]$ to evaluate to 1 if $a=b$ and to evaluate to 0 otherwise.

Given an instance $\tPi = (\tG = (\tV, \tE), \tSigma, \{\tC_e\}_{e \in \tE})$ of $\delta$-balanced 2-\CSP\ from \Cref{thm:balanced-2csp}. We create an GapMaxMin-2-\CSP$_q$\ instance $(\Pi = (G = (V, E), \Sigma, \{C_e\}_{e \in E}), \psi_s, \psi_t)$ as follows:\end{sloppypar}
\begin{itemize}
\item Let $G = \tG$.
\item Let $\Sigma = (\tSigma \cup \{\sigma^*\}) \times \{0, 1\}$ where $\sigma^*$ is a new character and $\Sigma$ denote two copies of $(\tSigma \cup \{\sigma^*\})$, indexed by the second coordinate. Moreover, for any $\sigma=(\tsigma,a)\in \Sigma$ where $\tsigma\in\tSigma\cup \{\sigma^*\}$ and $a\in\{0,1\}$, we denote by $\sigma_1:=\tsigma$ and $\sigma_2:=a$. 
\item For $e = (u, v) \in E$, we define $C_e : \Sigma \times \Sigma\to \{0,1\}$ as follows:
\begin{align*}
\forall (\sigma^u, \sigma^v)\in \Sigma\times \Sigma,\ \  C_e(\sigma^u, \sigma^v) := 
\begin{cases}
\tC_e(\sigma^u_1, \sigma^v_1) & \text{ if } (\sigma^u_1, \sigma^v_1) \in \tSigma\times \tSigma, \\
\ind[\sigma^u_2 , \sigma^v_2]&\text{ otherwise.}
\end{cases}
\end{align*}
\item Finally, for all $v \in V$, we define $\psi_s(v) := (\sigma^*, 0)$ and $\psi_t(v) := (\sigma^*, 1)$.
\end{itemize}

\paragraph{Completeness.} Suppose that $\val(\tPi) = 1$. That is, there exists an assignment $\tpsi^*$ that satisfies all constraints in $\tPi$. Let $\psi^*_0$ and $\psi^*_1$ be defined by $\psi^*_0(v) := (\tpsi^*(v), 0)$ and $\psi^*_1(v) := (\tpsi^*(v), 1)$ for all $v \in V$. Let $\bpsi$ be a concatenation of any direct sequence from $\psi_s$ to $\psi^*_0$, from $\psi^*_0$ to $\psi^*_1$ and from $\psi^*_1$ to $\psi_t$. It is simple to see that $\val_{\Pi}(\bpsi) = 1$ as desired.

\paragraph{Soundness.} Suppose that $\val(\tPi) < \delta$. Consider any $\bpsi = (\psi_0 = \psi_s, \dots, \psi_p = \psi_t) \in \bPsi(\psi_s \lrsg \psi_t)$ for some $p\in\mathbb{N}$. Let $n_i = |\{v \in V \mid \psi_i(v)_2 = 0\}|$ for all $i = 0, \dots, p$. Since $n_1 = 0, n_p = n$ and $|n_i - n_{i-1}| \leq 1$, there must be $i^*$ such that $n_{i^*} = \lceil n/2 \rceil$. Consider $\psi_{i^*}$. Let $\tpsi:\tV\to\tSigma$ be such that, for all $v \in V$, $\tpsi(v) = \psi_{i^*}(v)_1$ if $\psi_{i^*}(v)_1 \in \tSigma$ and $\tpsi(v)$ can be set arbitrarily otherwise. Furthermore, let $V_0 := \{v \in V \mid \psi_{i^*}(v)_2 = 0\}$ and $V_1 := \{v \in V \mid \psi_{i^*}(v)_2 = 1\}$. We have
\begin{align*}
\val_{\Pi}(\bpsi)\le \val_{\Pi}(\psi_{i^*})
&= \underset{{e = (u, v) \sim E}}{\E}\left[C_e(\psi_{i^*}(u), \psi_{i^*}(v))\right] \\
&\leq \underset{{e = (u, v) \sim E}}{\E}\left[\tC_e(\tpsi(u), \tpsi(v)) + \ind[\sigma^u_2 , \sigma^v_2]\right] \\
&\leq \val_{\tPi}(\tpsi) + \frac{|E[V_0] + E[V_1]|}{|E|} \\
&< \delta + (1/2 + \delta) = 1/2 + \eps,
\end{align*}
where the last inequality follows from $\val(\tPi) < \delta$ and that $\tPi$ is $\delta$-balanced.
\end{proof}

\subsubsection[Hardness of Balanced Gap-2-CSP]{Hardness of Balanced Gap-2-\CSP$_q$}
\label{subsec:balance}

In this subsubsection, we provide a short proof of \Cref{thm:balanced-2csp}. We note that this result seems to be folklore in literature. However, since we are not aware of the result stated exactly in this form, we show how to derive it from an explicitly stated result in~\cite{Moshkovitz14} for completeness. 

\paragraph{Additional Preliminaries.} To state this result, we some additional definitions.

For a distribution $P$ and a possible outcome $x$, we write $P(x)$ to denote the probability that the outcome is $x$. For a set $S$, we write $P(S)$ to denote $\sum_{x \in S} P(S)$. The \emph{total variation (TV) distance} between two distributions $P, Q$ is defined as $d_{TV}(P, Q) := \frac{1}{2} \sum_{x} |P(x) - Q(x)| = \max_S P(S) - Q(S)$.
The \emph{min-entropy} $P$ is defined to be $H_{\infty}(P) := \min_x \log(1/P(x))$. We write $U_S$ to denote the uniform distribution on a set $S$. 

For a bipartite graph $G=(X\uplus Y, E)$ and a distribution $P_X$ on $X$, let $G \circ P_X$ denote the distribution of sampling $x \sim P_X$ and then picking a uniformly random neighbor $y$ of $x$ in $G$.

An \emph{$(\delta, \eps)$-extractor graph} is a bi-regular bipartite graph $G=(X\uplus Y, E)$ that satisfies the following: For any distribution $P_X$ over $X$ with $H_{\infty}(P) \geq \log(\delta|X|)$, we have $d_{TV}(G \circ P_X, U_Y) \leq \eps$, where $U_Y$ is the uniform distribution over $Y$.

Moshkovitz~\cite{Moshkovitz14} gave a transformation from any 2-\CSP\ instance on arbitrary bi-regular graphs to one which is a good extractor, while preserving the value of the instance. This gives the following hardness of 2-\CSP{}s on extractor graphs.

\begin{theorem}[\cite{Moshkovitz14}] \label{thm:fortification}
For any constants $\gamma, \delta > 0$, there exists $q \in \N$ such that $\gapcsp{1}{\delta}{q}$ is \NP-hard even when the constraint graph is an $(\gamma, \gamma^2)$-extractor graph.
\end{theorem}

\paragraph{From Extractor to Balancedness.}
Given \Cref{thm:fortification}, it suffices for us to show that a good extractor graph also satisfies balancedness (\Cref{def:balance}). Before we show this, it will be helpful to state the following lemma, which is analogous to the ``expander mixing lemma'' but for extractors.

\begin{lemma} \label{lem:mixing}
Let $G = (X, Y, E)$ be any $(\gamma, \eps)$-extractor graph. Then, for any $X' \subseteq X, Y' \subseteq Y$, we have
\begin{align*}
\left|\frac{|E[X', Y']|}{|E|} - \frac{|X'||Y'|}{|X||Y|}\right| \leq \gamma + \eps
\end{align*}
\end{lemma}

In fact, an almost identical lemma was already shown in~\cite{vadhan-survey}, as stated below. The only difference is this version requires the size of $X'$ to be sufficiently large\footnote{In fact, Vadhan~\cite{vadhan-survey} proves ``if and only if'' statement but for $|X'| = \gamma |X|$. However, it is clear that the forward direction holds for any $|X'| \geq \gamma |X|$}.

\begin{lemma}[{\cite[Proposition 6.21]{vadhan-survey}}]\label{lem:mixing-vadhan}
Let $G = (X, Y, E)$ be any $(\gamma, \eps)$-extractor graph. Then, for any $X' \subseteq X, Y' \subseteq Y$ such that $|X'| \geq \gamma |X|$, we have
\begin{align*}
\left|\frac{|E[X', Y']|}{|E|} - \frac{|X'||Y'|}{|X||Y|}\right| \leq \eps
\end{align*}
\end{lemma}

Our version of the lemma follows almost trivially from the one above, as stated below.

\begin{proof}[Proof of \Cref{lem:mixing}]
Consider two cases based on the size of $X'$. If  $|X'| \geq \gamma |X|$, then this follows directly from \Cref{lem:mixing-vadhan}. Otherwise, if $|X'| < \gamma |X|$, we have both $\frac{|X'||Y'|}{|X||Y|} < \gamma$ and $\frac{|E[X', Y']|}{|E|} \leq \frac{|E[X', Y]|}{|E|} < \gamma$ (where the latter is from bi-regularity). Thus, we have  $\left|\frac{|E[X', Y']|}{|E|} - \frac{|X'||Y'|}{|X||Y|}\right| < \gamma$.
\end{proof}


From the above, we can conclude that any good extractor satisfies balancedness:

\begin{lemma} \label{lem:ext-balanced}
Any $(\gamma, \eps)$-extractor graph is $4(\gamma + \eps)$-balanced.
\end{lemma}

\begin{proof}
Consider any partition of $V = X \uplus Y$ into two balanced parts $V_1, V_2$ (i.e. such that $||V_1| - |V_2|| \leq 1$. Let $X_1 := V_1 \cap X, X_2 := V_2 \cap X, Y_1 := V_1 \cap Y$ and $Y_2 := V_2 \cap Y$. Now, we have
\begin{align*}
\frac{|E[V_1]| + |E[V_2]|}{|E|} &= \frac{|E[X_1, Y_1]|}{|E|} + \frac{|E[X_2, Y_2]|}{|E|} \\
 &\leq 2(\gamma + \eps) + \frac{|X_1||Y_1| + |X_2||Y_2|}{|X||Y|} & \text{(\Cref{lem:mixing})} \\
 &= 2(\gamma + \eps) + \frac{1}{2} + \frac{1}{2} \cdot \frac{(|X_1| - |X_2|)(|Y_1| - |Y_2|)}{|X| |Y|} \\
 &\leq 2(\gamma + \eps) + \frac{1}{2},
\end{align*}
where the last inequality follows from the fact that $||V_1| - |V_2|| \leq 1$ (which implies that $|X_1| - |X_2|, |Y_1| - |Y_2|$ cannot be both positive).
\end{proof}

Proof of \Cref{thm:balanced-2csp} is now complete by just combining the above results.

\begin{proof}[Proof of \Cref{thm:balanced-2csp}]
This follows immediately from \Cref{thm:fortification} with $\gamma = \delta/8$ since \Cref{lem:ext-balanced} asserts that any $(\gamma, \gamma^2)$-extractor graph is $\delta$-balanced.
\end{proof}

We end this subsection by noting that, when the constraint graph is the complete bipartite graph, then the instance is 0-balanced. This corresponds to the so-called \emph{free games}, which admits a PTAS for constant alphabet size $q$ but becomes hard to approximate when $q$ is large~\cite{AaronsonIM14,ManurangsiR17}. Such a hardness result is weaker than~\Cref{thm:balanced-2csp} in two ways: $q$ has to be super constant and the hardness is only under the Exponential Time Hypothesis (ETH)~\cite{IP01,IPZ01}.

\subsection[NP-Hardness of GapMinMax-2-CSP]{\NP-Hardness of GapMinMax-2-\CSP$_q$}

In this subsection, we will prove our (nearly) tight \NP-hardness of approximation of   GapMinMax-2-\CSP$_q$\ (\Cref{thm:minmax-csp-np-hardness}).

To do so, we will need the following hardness of Gap-2-\CSP$_q$\ in terms of partial assignment. Note that this can be easily derived from taking any  \PCP\ that reads $O_\delta(1)$-bits with $\delta$-soundness and plug it into the FGLSS reduction~\cite{FGLSS96}.
(This is usually stated in terms of the hardness of Maximum Clique problem, but it can be stated in the form below.)

\begin{theorem}[\cite{ALMSS98,FGLSS96}] \label{thm:csp-part}
For any $\delta > 0$, there exists $q \in \N$ such that, it is \NP-hard, given a Gap-2-\CSP$_q$ instance $\Pi = (G = (V, E), \Sigma, \{C_e\}_{e \in E})$, to distinguish between $\val(\Pi) = 1$ or $\maxpar(\Pi) < \delta \cdot |V|$.
\end{theorem}

We can now give a reduction from the hard Gap-2-\CSP$_q$\ instance above to prove \Cref{thm:minmax-csp-np-hardness} in a similar (but slightly simpler) manner as in the proof of \Cref{thm:maxmin-csp-np-hardness}.

\begin{proof}[Proof of \Cref{thm:minmax-csp-np-hardness}]
Given an instance $\tPi = (\tG = (\tV, \tE), \tSigma, \{\tC_e\}_{e \in \tE})$ of 2-\CSP\ from \Cref{thm:csp-part}. We assume w.l.o.g. that $\tG$ is a complete graph, i.e. $\tE = \binom{\tV}{2}$, as we can add trivial constraints over the non-edges in $\tE$ without changing the $\minlab$ value. We create an 2-\CSP\ Reconfiguration instance $(\Pi = (G = (V, E), \Sigma, \{C_e\}_{e \in E}), \psi_s, \psi_t)$ exactly as in the proof of \Cref{thm:maxmin-csp-np-hardness}.

\paragraph{Completeness.} Suppose that $\val(\tPi) = 1$. As shown in the proof of \Cref{thm:maxmin-csp-np-hardness}, there exists a reconfiguration assignment sequence $\bpsi = (\psi_0 = \psi_s, \dots, \psi_p = \psi_t)$ such that $\val_{\Pi}(\psi_i) = 1$ for all $i \in [p]$. We create a reconfiguration multi-assignment sequence $\bpsi' = (\psi_0, \psi'_1, \psi_1, \dots, \psi'_p, \psi_p)$ where we let $\psi'_i(v) = \{\psi_{i - 1}(v), \psi_i(v)\}$ for all $v \in V$ and $i \in [p]$. It is clear that this is a satisfying sequence and that $|\bpsi'| = |V| + 1$ as desired.

\paragraph{Soundness.} 
Let $\delta = \eps/2$.
We may assume w.l.o.g. that $|V| \geq 4/\eps$.
Suppose that $\maxpar(\Pi) < \delta \cdot |V|$. Consider any $\bpsi = (\psi_0 = \psi_s, \dots, \psi_p = \psi_t) \in \bPsi^{\SAT(\Pi)}(\psi_s \lrsg \psi_t)$. Let $i \in [p]$ be the smallest index\footnote{Such an index always exists since $\psi_t$ satisfies this condition.} for which $\psi_i(v) \ne \{(\sigma^*, 0)\}$ for all $v \in V$. Consider $\psi_{i - 1}$; there must exist $v_s$ such that $\psi_{i-1}(v_s) = \{(\sigma^*, 0)\}$. Due to the definition of $C_{(v_s, u)}$ and since $\psi_{i-1}$ satisfies $\Pi$, it must be the case that $\psi_{i-1}(u) \ne \{(\sigma^*, 1)\}$ for all $u \in V$. As a result, we also have that $\psi_i(u) \ne \{(\sigma^*, 1)\}$.

Now, consider $V_1 := \{v \mid |\psi_i(v)| = 1\}$. For each $v \in V_1$, let $\sigma_v$ denote the only element of $\psi(v)$. From the above paragraph, we must have $(\sigma_v)_1 \in \tSigma$ for all $v \in V_1$. Thus, we may define the $\psi': V \to \Sigma \cup \{\perp\}$ by
\begin{align*}
\psi'(v) =
\begin{cases}
(\sigma_v)_1 & \text{ if } v \in V_1, \\
\perp & \text{otherwise.}
\end{cases}
\end{align*}
Since $\psi$ satisfies $\Pi$, $\psi'$ is a satisfying partial assignment to $\tPi$. From our assumption that $\maxpar(\Pi) < \delta \cdot |V|$, we must have $|V_1| < \delta \cdot |V|$. As a result, we have
\begin{align*}
|\psi_i| \geq 2 \cdot |V \setminus V_1| + |V_1| = 2 \cdot |V| - |V_1| > (2 - \delta) \cdot |V| \geq (2 - \eps) \cdot (|V| + 1).
\end{align*}
This implies that $|\bpsi| > (2 - \eps) \cdot (|V| + 1)$ as claimed.
\end{proof}

\subsection{MinMax Set Cover Reconfiguration}

In this subsection, we will prove our (nearly) tight \NP-hardness of approximation of the  Set Cover Reconfiguration problem (\Cref{thm:minmax-setcover-np-hardness}).

It turns out that the classic reduction from Gap-2-\CSP$_q$\ to Set Cover of Lund and Yannakakis~\cite{LY94} also yields a gap preserving reduction from  the GapMinMax-2-\CSP$_q$ problem to the Set Cover Reconfiguration problem. This reduction was also used by Ohsaka~\cite{ohsaka2023gap}.
We summarize the properties of the reduction below.\footnote{Note that a different ``hypercube gadget'' reduction of Feige~\cite{Feige98} also have similar properties. See e.g.~\cite{ChalermsookCKLM20,SLM19} for description of this reduction in the non-reconfiguration setting in the MinLabel terminology.}

\begin{theorem}[\cite{LY94,ohsaka2023gap}]
There is a reduction that takes in as input a Gap-2-\CSP$_q$ instance $\Pi = (G = (V, E), \Sigma, \{C_e\}_{e \in E})$ and produces a Set Cover  instance $(S_{v,\sigma})_{v \in V, \sigma \in \Sigma}$ such that
\begin{itemize}
\item a multi-assignment $\psi: V \to \cP(\Sigma)$ satisfies $\Pi$ if and only if $\{S_{v, \sigma}\}_{v \in V, \sigma \in \psi(v)}$ is a set cover
, and,
\item $m = |V|$ and $|S_i| \leq |E| \cdot 2^q$ for all $i \in [m]$.
\end{itemize} 
Moreover, the algorithm runs in polynomial time in $(|V| + |E|) \cdot 2^q$.
\end{theorem}

Notice that the first property implies that each set cover reconfiguration sequence from $T_s := \{S_{v, \sigma}\}_{v \in V, \sigma \in \psi_s(v)}$ to $T_t := \{S_{v, \sigma}\}_{v \in V, \sigma \in \psi_t(v)}$ in the set cover instance has a one-to-one correspondence with a satisfying reconfiguration multi-assignment
sequence from $\psi_s$ to $\psi_t$ in $\Pi$ where the size is preserved. As a result, plugging the above into \Cref{thm:minmax-csp-np-hardness}, we immediately arrive at \Cref{thm:minmax-setcover-np-hardness}.

\section[Approximate Algorithm for GapMaxMin-2-\CSP]{Approximate Algorithm for GapMaxMin-2-\CSP$_q$}

In this section, we give the approximation algorithm for GapMaxMin-2-\CSP$_q$ (\Cref{thm:approx-algo-2csp}). Our main result is actually a structural theorem showing that, in any graph, we can find sequence of downward subsets of the vertices such that at most roughly half edges are cut by these sets:

\begin{theorem} \label{thm:sequence-balanced-full}
For any graph $G = (V, E)$ with $m$ edges, there exists a downward sequence $V = S_0 \supsetneq \dots \supsetneq S_n = \emptyset$ such that $\underset{i \in [n]}{\max}\ \left|E[S_i, V \setminus S_i]\right| \leq m/2 + 7m^{4/5}$. Furthermore, such a sequence can be computed in (randomized) polynomial time.
\end{theorem}

Note that it is not clear if the $\Theta(m^{4/5})$ additive factor is tight. The best known lower bound we are aware of is $m/2 + \Theta(\sqrt{m})$ which happens when we take an $n$-clique such that $n$ is odd.

We also note that \Cref{thm:sequence-balanced-full} is an improvement on a similar theorem in \cite{ohsaka2023gap} where $m/2$ is replaced with $3m/4$ (and different lower order term). Using a similar strategy, we can immediately get an approximation algorithm for GapMaxMin-2-\CSP$_q$, as formalized below.

\begin{proof}[Proof of \Cref{thm:approx-algo-2csp}]
Let $(\Pi = (G = (V, E), \Sigma, \{C_e\}_{e \in E}), \psi_s, \psi_t)$ be the input instance of GapMaxMin-2-\CSP$_q$. If $m \leq 10^6/\eps$, then use the exponential-time exact algorithm to solve the problem. Otherwise, use \Cref{thm:sequence-balanced-full} to first find a sequence $V = S_0 \supsetneq \dots \supsetneq S_n = \emptyset$ such that $\underset{i \in [n]}{\max}\ \left|E[S_i, V \setminus S_i]\right| \leq m/2 + 7m^{4/5}$. We define the direct reconfiguration assignment sequence $\psi_0, \dots, \psi_n$ by 
\begin{align*}
\psi_i(v) = 
\begin{cases}
\psi_s(v) &\text{ if } v \in S_i, \\
\psi_t(v) &\text{ otherwise.}
\end{cases}
\end{align*}
It is simple to see that for all $i \in [n]$, we have $\val_{\Pi}(\psi_i) \geq \frac{1}{|E|} \cdot \left(|E| - |E[S_i, V \setminus S_i]|\right) \geq 1/2 - 7/m^{1/5} \geq 1/2 - \eps$. This completes the proof.
\end{proof}

The rest of this section is dedicated to the proof of \Cref{thm:sequence-balanced-full}.

\subsection{Low-Degree Case}

At a high-level, it seems plausible that a random sequence satisfies this property with $1 - o(1)$ probability. However, this is not true: if our graph is a star, then, with probability $1 - \gamma$, the maximum cut size $\max_{i \in [n]} |E[S_i, V \setminus S_i]|$ is at least $1/2 + \Omega(\gamma)$ of the entire graph. The challenge in this setting is the high-degree vertex. Due to this, we start by assuming that the max-degree of the graph is bounded and prove the following:

\begin{lemma} \label{lem:balanced-rearrangement-low-deg}
For any graph $G = (V, E)$ with $m$ edges such that each vertex has degree at most $\Delta$, there exists a downward sequence $V = S_0 \supsetneq \dots \supsetneq S_n = \emptyset$ such that $\underset{i \in [n]}{\max}\ \left|E[S_i, V \setminus S_i]\right| \leq m/2 + \sqrt{m\Delta} + \Delta$. Furthermore, such a sequence can be computed in (randomized) polynomial time.
\end{lemma}

To prove \Cref{lem:balanced-rearrangement-low-deg}, we start by showing that the following lemma on graph partitioning, that roughly balances the degree and cuts half the edges.

\begin{lemma} \label{lem:balanced-partition}
For any graph $G = (V, E)$ with $m$ edges such that each vertex has degree at most $\Delta$, there exists a partition of $V$ into $V_1 \uplus V_2$ such that
\begin{align}
|E[V_1, V_2]| &\leq m/2 + \sqrt{m} \label{eq:small-bisect} \\
\sum_{v \in V_1} \deg_G(v)&\leq m + 2\sqrt{m\Delta},\quad \sum_{v \in V_2} \deg_G(v) \leq m + 2\sqrt{m\Delta}. \label{eq:balanced-degree}
\end{align}
Furthermore, such a partition can be computed in (randomized) polynomial time.
\end{lemma}

\begin{proof}
We will show that a random partition $V_1, V_2$ (where, for each $v \in V$, we pick $i(v)$ uniformly at random from $\{1, 2\}$ and let $v$ be in $V_{i(v)}$) satisfies the desired condition with probability\footnote{We can repeat this process to make the failure probability arbitrarily small.} 1/4. First, observe that $\E[|E[V_1, V_2]|] \geq m/2$. Meanwhile, we have
\begin{align*}
\E[|E[V_1, V_2]|^2] = \sum_{(u, v) \in E, (u', v') \in E} \Pr[\{i(u), i(v)\} = \{1, 2\}, \{i(u'), i(v')\} = \{1, 2\}]. 
\end{align*}
Notice that
\begin{align*}
\Pr[\{i(u), i(v)\} = \{1, 2\}, \{i(u'), i(v')\} = \{1, 2\}] =
\begin{cases}
1/4 & \text{ if } \{u, v\}  \ne \{u', v'\}, \\
1/2 & \text{ if } \{u, v\}  = \{u', v'\}.
\end{cases}
\end{align*}
Plugging this into the above, we have $\E[|E[V_1, V_2]|^2] = m^2/4 + m/4$. This means that $\Var(|E[V_1, V_2]|) \leq m/4$ and, thus, by Chebyshev's inequality, we have
\begin{align} \label{eq:cherbyshev-one}
\Pr[|E[V_1, V_2]| > m/2 + \sqrt{m}] < 1/4.
\end{align}
Meanwhile, we also have $\E[\sum_{v \in V_1} \deg_G(v)] = m$ and 
\begin{align*}
\Var\left(\sum_{v \in V_1} \deg_G(v)\right)
&= \Var\left(\sum_{v \in V} \deg_G(v) \cdot \ind[i(v) = 1]\right) \\
&= \sum_{v \in V} \Var\left(\deg_G(v) \cdot \ind[i(v) = 1]\right) \\
&\leq \sum_{v \in V} \deg_G(v)^2 \\
&\leq 2m\Delta,
\end{align*}
where the last inequality is due to the assumption that the maximum degree is at most $\Delta$.

Again, by Chebyshev's inequality, we have
\begin{align} \label{eq:cherbyshev-two}
\Pr\left[\left|\sum_{v \in V_1} \deg_G(v) - m\right| > 2\sqrt{m\Delta}\right] < 1/2.
\end{align}
Combining \eqref{eq:cherbyshev-one} and \eqref{eq:cherbyshev-two}, we can conclude that a random partition satisfies the two condition with probability at least 1/4 at desired.
\end{proof}

\Cref{lem:balanced-rearrangement-low-deg} can now be proved using the vertex set partitioning guaranteed by the above lemma and then building the sequence (to the left and to the right) using a ``greedy'' strategy.

\begin{proof}[Proof of \Cref{lem:balanced-rearrangement-low-deg}]
First, apply \Cref{lem:balanced-partition} to obtain a partition $V = V_1 \uplus V_2$ that satisfies \eqref{eq:small-bisect} and \eqref{eq:balanced-degree}. Let $n_1 = |V_1|$ and $S_{n_1} = V_2$. We then define the remaining $S_t$s using a ``greedy'' approach as follows.
\begin{itemize}
\item For $t = n_1 - 1, \dots, 0$, pick $v_t := \underset{v \in V \setminus S_{t+1}}{\argmin}\ |E[S_{t + 1} \cup \{v\}, V \setminus (S_{t + 1} \cup \{v\})]|$. Then, let $S_t = S_{t + 1} \cup \{v_t\}$.
\item For $t = n_1 + 1, \dots, n$, pick $v_t := \underset{v \in S_{t-1}}{\argmin}\ |E[S_{t - 1} \setminus \{v\}, V \setminus (S_{t - 1} \setminus \{v\})]|$. Then, let $S_t = S_{t - 1} \setminus \{v_t\}$.
\end{itemize}
Suppose for the sake of contradiction that $\underset{t \in [n]}{\max}\ |E[S_t, V \setminus S_t]| > m/2 + \sqrt{m\Delta} + \Delta$. Let $t^* = \underset{t \in [n]}{\argmax}\ |E[S_t, V \setminus S_t]|$. Note that \eqref{eq:small-bisect} implies that $t^* \ne n_1$. Thus, we must have that either $t^* < n_1$ or $t^* > n_1$. Since the two cases are symmetric, we assume w.l.o.g. that $t^* < n_1$.

Since the degree of every vertex is at most $\Delta$, we also have that $|E[S_{t^* + 1}, V \setminus S_{t^* + 1}]| > m/2 + \sqrt{m\Delta}$. Note that $(V \setminus S_{t^* + 1}) \subseteq (V \setminus S_{n_1}) = V_1$. Thus, \eqref{eq:balanced-degree} implies that $\sum_{v \in (V \setminus S_{t^*+1})} \deg_G(v) \leq m + 2\sqrt{m\Delta} < 2 \cdot |E[S_{t^* + 1}, V \setminus S_{t^* + 1}]|$. This means that there exists $v' \in (V \setminus S_{t^*+1})$ such that $v$ has more edges to $S_{t^* + 1}$ than that within $(V \setminus S_{t^*+1})$. In other words, $|E[S_{t^* + 1} \cup \{v'\}, V \setminus (S_{t^* + 1} \cup \{v'\})]| < |E[S_{t^* + 1}, V \setminus S_{t^* + 1}]|$. By the algorithm's greedy choice of $v_{t^*}$, we must also have $|E[S_{t^*}, V \setminus S_{t^*}]| < |E[S_{t^* + 1}, V \setminus S_{t^* + 1}]|$. However, this contradicts with our choice of $t^*$.
\end{proof}

\subsection{Handling High-Degree Vertices}

We now prove \Cref{thm:sequence-balanced-full}. At a high level, this is done by applying the previous subsection's result on the low-degree vertices to get a sequence of sets on those. We then interleave the high-degree vertices using a ``greedy'' strategy where we move a high-degree vertex out of the set only when it decreases the cut size w.r.t.\ edges to the low-degree vertices. We formalize and analyze this strategy below.

\begin{proof}[Proof of \Cref{thm:sequence-balanced-full}]
Let $\Delta := 2m^{3/5}$. We partition $V$ into $V^{> \Delta}$ and $V^{\leq \Delta}$. In $V^{> \Delta}$ (resp.\ $V^{\leq \Delta}$) we have those vertices with degree more than  $\Delta$ (resp.\ vertices with degree at most $\Delta$). For brevity, let $E_1 := E[V^{\leq \Delta}],\ E_2 := E[V^{> \Delta}, V^{\leq \Delta}],\ E_3 = E[V^{> \Delta}]$ and $m_1 = |E_1|,\ m_2 = |E_2|,\ m_3 = |E_3|$. Note that $E_1 \uplus E_2 \uplus E_3$ is a partition of $E$.


First, we invoke \Cref{lem:balanced-rearrangement-low-deg} to get a sequence $V^{\leq \Delta} = \tS_0 \supsetneq \cdots \supsetneq \tS_{n_1} = \emptyset$ such that for all $\ttt \in \{0, \dots, n_1\}$ we have $|E[\tS_{\ttt}, V^{\leq \Delta} \setminus \tS_{\ttt}]| \leq m_1/2 + \sqrt{m_1\Delta} + \Delta$. For the next step, it will be convenient to define $\tv_{\ttt}$ as the only vertex in $\tS_{\ttt} \setminus \tS_{\ttt + 1}$ for all $\ttt \in \{0, \dots, n_1 - 1\}$.
Then, we construct the full set sequence through the following procedure:
\begin{enumerate}
\item Let $j = 0$ and define $S_j := V$.
\item For $i = 0, \dots, n_1 - 1$:
\begin{enumerate}
\item[2.1.] \textbf{Low-Degree Move:} Define $S_{j+1} := S_j \setminus \{\tv_i\}$. (Note that we have $|S_{j+1} \cap V^{\leq \Delta}| = \tS_{i+1}$) Then, increase $j$ by one. 
\item[2.2.] \textbf{High-Degree Correction:} Next, while there exists $v \in V^{> \Delta} \cap S_j$ such that $|E[v, \tS_i]| \leq |E[v, V \setminus \tS_i]|$:
\begin{itemize}
\item Define $S_{j + 1} := S_j \setminus \{v\}$. Then, increase $j$ by one. 
\end{itemize}
\end{enumerate}
\end{enumerate}
Note that, since $\Delta \geq 2\sqrt{m}$, we have $|V^{> \Delta}| < 2m / \Delta \leq \Delta / 2$. This means that every high-degree vertex $v \in V^{>\Delta}$ has at least as many edges to $V^{\leq \Delta}$ as it has within $V^{>\Delta}$. This implies that, after the last low-degree move (i.e. $i = n_1 - 1$), the high-degree correction will move all vertices out of $S$ as desired.

Next, we will argue that $\underset{t \in [n]}{\max}\ \left|E[S_t, V \setminus S_t] \cap (E_1 \cup E_2)\right| \leq  m/2 + \sqrt{m\Delta} + 2\Delta$. Before we do so, note that since $m_3 \leq |V^{> \Delta}|^2/2 < 2(m/\Delta)^2$. This implies that $\underset{t \in [n]}{\max}\  \left|E[S_t, V \setminus S_t]\right| \leq  m/2 + 2(m/\Delta)^2 + \sqrt{m\Delta} + 2\Delta$, which is at most $m/2 + 7m^{4/5}$ as claimed.

To bound $\underset{t \in [n]}{\max}\ \left|E[S_t, V \setminus S_t] \cap (E_1 \cup E_2)\right   |$, first observe that high-degree correction never increases $|E[S_t, V \setminus S_t] \cap (E_1 \cup E_2)|$. Thus, it suffices to argue that $|E[S_t, V \setminus S_t] \cap (E_1 \cup E_2)| \leq m/2 + \sqrt{m\Delta} + 2\Delta$ immediately after a low-degree move. Since this is a low-degree move, we have 
\begin{align*}
|E[S_t, V \setminus S_t] \cap (E_1 \cup E_2)| \leq |E[S_{t - 1}, V \setminus S_{t - 1}] \cap (E_1 \cup E_2)| + \Delta.
\end{align*} 
Moreover, since $S_{t - 1}$ is a result after all possible high-degree corrections, for every $v \in V^{> \Delta}$, at most half of its edges to $V^{\leq \Delta}$ belong to the cut. As a result, $$|E[S_{t - 1}, V \setminus S_{t - 1}] \cap E_2| \leq m_2 / 2.$$
Meanwhile, from our choice of $\tS_0, \dots, \tS_{n_1}$, the number of edges in $E_1$ that are cut is at most $m_1 / 2 + \sqrt{m_1\Delta} + \Delta$, i.e.,
\begin{align*}
|E[S_{t - 1}, V \setminus S_{t - 1}] \cap E_1| \leq m_1 / 2 + \sqrt{m_1\Delta} + \Delta.
\end{align*}
Combining these three inequalities, we have
$$|E[S_t, V \setminus S_t]| + m_2 / 2 + m_1 / 2 + \sqrt{m_1\Delta} + \Delta + \Delta \leq m/2 + \sqrt{m/\Delta} + 2\Delta,$$ which concludes our proof.
\end{proof}

\section{Conclusion and Open Questions}
\begin{sloppypar}
We positively resolved the Reconfiguration Inapproximability Hypothesis~(\RIH)~\cite{Ohsaka23prev}, which in turn shows that a host of reconfiguration problems are \PSPACE-hard even to approximate within some constant factor. Meanwhile, we prove tight \NP-hardness of approximation results for GapMaxMin-2-\CSP$_q$ and Set Cover Reconfiguration.
\smallskip

Subsequent to \cite{KM24}, Hirahara and Ohsaka proved the \emph{optimal} inapproximability threshold for the Minmax Set Cover Reconfiguration problem, showing that it is \PSPACE-hard to approximate within a factor of $2-\tfrac{1}{\mathrm{polyloglog}\,N}$ (hence $2-o(1)$), matching the known $2$-approximation and therefore pinning down the threshold at~$2$~\cite{HiraharaO24setcover}. In contrast, the tight \PSPACE-hardness threshold for GapMaxMin-2-\CSP$_q$ remains open.
\smallskip

In light of \Cref{thm:main_grw}, a direct route to closing the gap for GapMaxMin-2-\CSP$_q$ is to establish concrete and sharp bounds on the trade-offs between the query complexity and the soundness gap of \PCPP\ (assignment testers). On the other hand, a negative answer---by placing the gap version in a complexity class believed to be a strict subset of \PSPACE\ (e.g., $\Sigma^P_2$)---would also be very interesting (indeed, arguably more so than a \PSPACE-hardness result).

Apart from the aforementioned question, there are several other more technical and specific problems. For example, what are the best $\eps$ in terms of $n$ (and $q$) that we can get in our \NP-hardness results? Specifically, \Cref{thm:maxmin-csp-np-hardness} requires the alphabet size $q$ to grow as $\eps \to 0$. Can we get rid of such a dependency (or show that this is not possible by giving an approximation algorithm)?
\end{sloppypar}

\bibliographystyle{alpha}
\bibliography{merge.bib}

\end{document}